\documentclass[11pt, a4paper]{article}
\usepackage{amsmath}
\usepackage{graphicx}
\usepackage[cm]{fullpage}
\usepackage{natbib}
\usepackage{enumerate}
\usepackage{color, soul}

\newtheorem{theorem}{Theorem}[section]
\newtheorem{lemma}[theorem]{Lemma}
\newenvironment{proof}[1][Proof]{\begin{trivlist}
\item[\hskip \labelsep {\bfseries #1}]}{\end{trivlist}}

\newcommand{\qed}{\nobreak \ifvmode \relax \else
      \ifdim\lastskip<1.5em \hskip-\lastskip
      \hskip1.5em plus0em minus0.5em \fi \nobreak
      \vrule height0.75em width0.5em depth0.25e                                                                                                        m\fi}

\title{Post-selection point and interval estimation of signal sizes in Gaussian samples}

\author{Stephen Reid$^1$, Jonathan Taylor$^1$ and Robert Tibshirani$^2$}
\date{\normalsize $^1$ Department of Statistics, Stanford University, $^2$ Departments of Health, Research \& Policy and Statistics, Stanford University}

\begin{document}

\maketitle

\begin{abstract}
We tackle the problem of the estimation of a vector of means from a single vector-valued observation $y$. Whereas previous work reduces the size of the estimates for the largest (absolute) sample elements via shrinkage (like James-Stein) or biases estimated via empirical Bayes methodology, we take a novel approach. We adapt recent developments by \citet{LeeSun2TaylorPostSel} in post selection inference for the Lasso to the orthogonal setting, where sample elements have different underlying signal sizes. This is exactly the setup encountered when estimating many means. It is shown that other selection procedures, like selecting the $K$ largest (absolute) sample elements and the Benjamini-Hochberg procedure, can be cast into their framework, allowing us to leverage their results. Point and interval estimates for signal sizes are proposed. These seem to perform quite well against competitors, both recent and more tenured.

Furthermore, we prove an upper bound to the worst case risk of our estimator, when combined with the Benjamini-Hochberg procedure, and show that it is within a constant multiple of the minimax risk over a rich set of parameter spaces meant to evoke sparsity.
\end{abstract}


\section{Introduction}
\label{sec:introduction}
In our world of big data, large scale studies and multiple testing, we are frequently confronted with the problem of estimating a vector of means $\mu$ form a single vector of observed data $y$. We choose to model the setup in the following way:

\begin{equation}
Y_i \sim N(\mu_i, \sigma^2)
\label{eqManyMeansModel}
\end{equation}
for $i = 1, 2, ..., n$. 

Here we have, in principle, a different effect size ($\mu_i$) for every data point ($Y_i$), which contrasts with the more classical problem of estimating a single effect size (mean) $\mu$ given a sample of size $n$ from the same population. The maximum likelihood estimator for each $\mu_i$ is $Y_i$ and the average mean squared error (MSE) over the entire $\mu$-vector is $\sigma^2$.

Applications often focus on a smaller part of $\mu$. An implicit assumption is one of sparsity, with many effect sizes $\mu_i = 0$. The goal is to identify those select few non-zero effects bobbing serendipitously like some precious flotsam in a sea of nulls.

Useful tools for the detection of these non-zero effects have been developed. Most of them --- be it those controlling the family-wise error rate (FWER), false discovery rate (FDR) or some other method of dubious validity --- propose as non-zero those effects corresponding to sample elements with the largest (absolute) size. 

Intuitive as this may be, these large (absolute) sample elements are poor estimators for their underlying effect size. A sample element is large for two reasons: its effect size is large and its random error pushed it away from zero. We experience a selection bias, whereby the largest sample elements tend to be those lucky few with large positive error terms, making them upwardly biased for their effect sizes.

As a demonstration of this selection bias or ``winner's curse'', consider the setup in \eqref{eqManyMeansModel} with $\mu_i = 0$ for all $i$ and $\sigma^2 = 1$ and let 
\[
|Y|_{(1)} \geq |Y|_{(2)} \geq ... \geq |Y|_{(n)}
\]
be the order statistics of the absolute values of the sample. Also, let $r(k)$ define the permutation such that $|Y_{r(k)}| = |Y|_{(k)}$. By a second order Taylor expansion, we can approximate

\begin{align*}
  E[|Y|_{(k)}] &= E\left[\Phi^{-1}\left(\frac{U_{(k)} - 1}{2}\right)\right] \\
  &\approx \Phi^{-1}\left(\frac{-k}{2(n+1)}\right)\left[1 + \frac{k(n-k+1)}{8(n+1)^2(n+2)\phi\left(\Phi^{-1}\left(\frac{-k}{2(n+1)}\right)\right)}\right]
\end{align*}
since $U_{(k)}$, the $k^{th}$ order statistic of a sample of size $n$ from a Uniform$(0, 1)$ distribution, has a Beta$(n-k+1, k)$ distribution. Figure~\ref{figWinnersCurse} shows some implications of this result for a sample of size $n = 100$. The left panel plots the (absolute) order statistics of a sample against rank, with the approximate expected value curve as reference. Notice how steeply the curve decreases as we move away from larger order statistics. The largest order statistics are considerably upwardly biased for their true effect sizes.

\begin{figure}[htb]
  \centering
  \includegraphics[width=120mm]{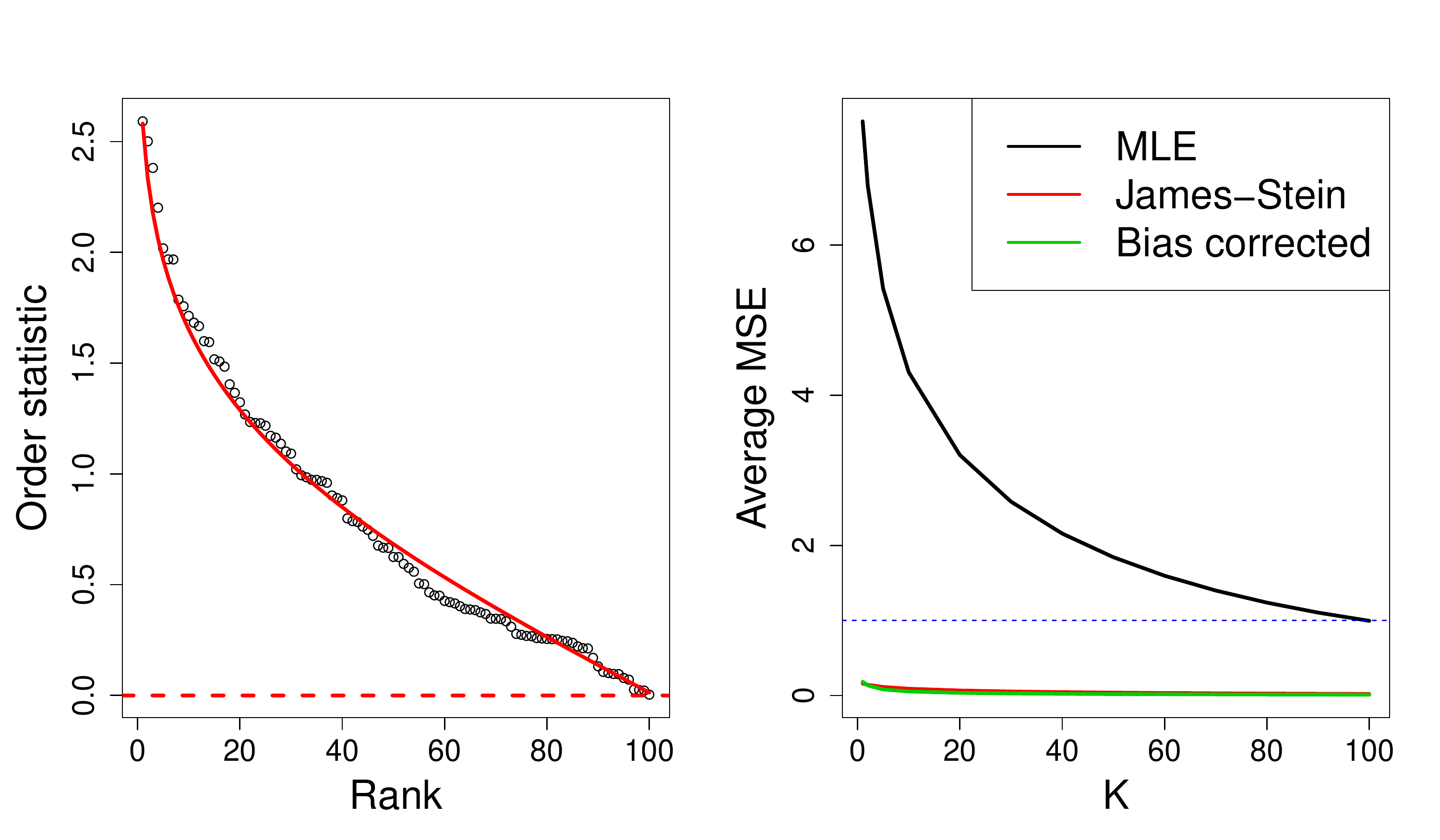}
  \caption{\emph{Left panel: Absolute order statistics from N(0,1) sample plotted as function of rank (k) with approximate expected value as computed from the second order Taylor expansion in red. Broken horizonal line shows true effect size (0). Right panel: Average MSE over first K (absolute) order statistics from S = 1000 simulated samples from N(0,1) distribution as function of K. Blue line at $\sigma^2 = 1$, the MSE over the full $\mu$-vector, for reference.}}
  \label{figWinnersCurse}
\end{figure}

The right panel of Figure~\ref{figWinnersCurse} (black curve) plots the average over $S = 1000$ simulations of the quantity

\[
MSE(K) = \frac{1}{K}\sum_{k = 1}^K(Y_{r(k)} - \mu_{r(k)})^2
\]
which is the mean squared error for effect sizes, considering only the first $K$ largest sample elements. Notice that the curve decreases (suggesting that we do worst for the very largest order statistics) and only attains the proven MSE of $\sigma^2 = 1$ at $K = n$. Stopping before $K = n$ and considering only the first $K$ order statistics, leads to very poor estimates of the effect sizes.

Two other MSE curves are plotted. The red curve is the average MSE for the James-Stein estimator:
\[
MSE^{JS}(K) = \frac{1}{K}\sum_{k = 1}^K(\hat{\mu}^{JS}_{r(k)} - \mu_{r(k)})^2
\]
where
\[
\hat{\mu}^{JS} = \left(1 - \frac{n - 2}{||Y||_2^2}\right)Y
\]

The green curve plots the MSE for the bias-corrected order statistics, subtracting the approximate expected value computed above:
\[
MSE^{BC}(K) = \frac{1}{K}\sum_{k = 1}^K\left({\rm sign}(Y_{r(k)})\left(|Y|_{(k)} - E[|Y|_{(k)}]\right) - \mu_{r(k)}\right)^2
\]

Notice how these two estimators improve vastly on the raw order statistics at all values of $K$. These two estimators cannot be used for our purposes: James-Stein, because it always shrinks to zero and bias-correction, because it relies on the fact that all effect sizes are the same. It is clear that the raw order statistics, when considered only a few at a time, need to be adjusted before they can be deemed good estimates of the effect size.

Many tools have been developed to counter this selection bias and some of them are reviewed in the next section. Our contribution adapts the recent work on post-selection inference by \citet{TTTPostSel} and \citet{LeeSun2TaylorPostSel} to the estimation of effect sizes. These papers (the latter in particular) suggest that, conditional on having selected the $K$ largest (absolute) order statistics to represent our non-zero effects, the distribution of $Y_{r(1)}, Y_{r(2)}, ..., Y_{r(K)}$ is that of a Gaussian sample with variance $\sigma^2$ truncated to the interval $(-\infty; -|Y|_{(K + 1)}) \cup (|Y|_{(K + 1)}; \infty)$.

The density function of such a Gaussian random variable with parameters $\mu$ and $\sigma^2$, truncated to $(-\infty; a) \cup (b; \infty)$ is
\[
 f_{\mu, \sigma^2, a, b}(x) = \frac{\frac{1}{\sigma}\phi\left(\frac{x - \mu}{\sigma}\right)}{\Phi\left(\frac{a - \mu}{\sigma}\right) + 1 - \Phi\left(\frac{b - \mu}{\sigma}\right)}
\]
leading to the maximum (conditional) likelihood estimator of $\mu_{r(k)}$ as the solution to (assuming known $\sigma$)
\[
Y_{r(k)} = \mu + \sigma\frac{\phi\left(\frac{|Y|_{(K+1)} - \mu}{\sigma}\right) - \phi\left(\frac{-|Y|_{(K+1)} - \mu}{\sigma}\right)}{\Phi\left(\frac{-|Y|_{(K+1)} - \mu}{\sigma}\right) + 1 - \Phi\left(\frac{|Y|_{(K+1)} - \mu}{\sigma}\right)}
\]

We show in a simulation study that these estimators perform commendably when estimating the effect sizes of the largest $K$ effect sizes, often competing with and sometimes outperforming some of the competitors. These are discussed next. Our truncated Gaussian estimator seems to perform best when true signals are sparse, with a mixture of different signal sizes.

Section~\ref{sec:previous_work} reviews some previous work in the field of many means estimation. We attempt to define a loose grouping of different types of estimators. Section~\ref{sec:post_sel_theo} discusses the post-selection framework of \citet{LeeSun2TaylorPostSel}. In this section, we show how their framework can be translated to the orthogonal setting, allowing us to use their results to obtain post-selection estimators for underlying signal sizes. A simulation experiment is presented in Section~\ref{sec:simulation}. Here we compare the estimation performance of our proposed estimator to that of a battery of previously proposed estimators. We consider estimation performance only over the set of sample elements selected as non-zero in some first phase selection procedure. Section~\ref{sec:risk} shows how the worst case risk of our estimator is bounded within a constant multiple of the minimax risk over a rich set of parameter spaces evoking sparsity of the underlying signal. Our analysis is transferred to confidence intervals in Section~\ref{sec:ci}, where we show how the proposed interval estimator addresses many of the issues encountered in post selection confidence interval construction. We conclude in Section~\ref{sec:conclude}.


\section{Previous work}\label{sec:previous_work}
Currently, there seem to be three strands of research on the problem of dealing with selection bias and post-selection effect size estimation: an empirical Bayes approach, a classical resampling approach and one of thresholding. There is some overlap between these strands.

\subsection{Empirical Bayes and density estimation}
\citet{EfronTweedie} provides an elegant framework for considering selection bias. Using a Bayesian approach, he postulates a prior $g(\mu)$ for the effect sizes and then, assuming $Y|\mu \sim N(\mu, \sigma^2)$, quotes Tweedie's formula for the posterior mean:
\[
E(\mu|x) = \mu + \sigma^2\frac{f^\prime(x)}{f(x)}
\]
where
\[
f(x) = \int \frac{1}{\sigma} \phi\left(\frac{x - \mu}{\sigma}\right)g(\mu)\,d\mu
\]
is the marginal distribution of $Y$.

He notes that a debiased estimate thus requires a good estimate of the marginal density. This observation allows the vast literature on density estimation to be applied to the selection bias problem. The literature is too large to do it justice here, but we mention some recent developments.

\citet{EfronTweedie} suggests the form
\[
f(x) = \exp\left[\sum_{j = 0}^J\beta_jx^j\right]
\]
which represents a $J$ parameter exponential family with $\beta_0$ as normalising constant. Lindsay's method allows the calculation of $\hat{\beta}$ using Poisson generalized linear model (GLM) software. The method requires a choice of $J$.

\citet{MuralidharanMix} suggests a mixture density prior
\[
 g(\mu) = \sum_{j = 0}^{J-1}\pi_jg_j(\mu)
\]
where the mixture component densities $g_j(\mu)$ are chosen so as to lead to tractable marginals and posteriors (say as conjugate priors in an exponential family setup) and $g_0(\mu)$ usually corresponds to a point mass at zero. With tractable, closed form marginals, the parameters $\pi$ and whatever parameters are associated with $g_j$ are estimated, giving a fully parameterized form for the posterior density and hence the posterior mean.

Although both of these methods are transparent and their models fairly simple to fit, \cite{WagerDens} cites problems with these density estimators in the tails of the estimated marginal distribution, precisely where we require good estimates for our debiased effect size estimate. He proposes a geometric approach to the estimation of $f(x)$. He finds the the closest distribution to the empirical distribution of the observed $Y$s, under the constraint that it be in the form of a convolution $\phi * g$. A rapid quadratic programming algorithm is proposed which provides provably good estimates comparing favourably to the gold-standard, albeit more computationally intensive, maximum likelihood estimators of \citet{JiangZhangMLDens}. 

These authors use a generalized maximum likelihood approach to find an optimal prior distribution. They do not restrict the prior to be in some class, like mixture densities or specific conjugate priors, for example. Rather, they treat the prior as an unknown to be estimated via maximum likelihood on the marginal distribution of the observed data. Their estimator for the prior is a discrete one, with mass estimated at a grid of equally spaced points. A version of the EM algorithm is used to obtain this estimate. Once obtained, signal size estimates are obtained as posterior weighted means of the observations.

\citet{JohnSilver2004} provide a slight variation on the theme. Still using the Bayes methodology, these use the a prior with mass at 0:
\[
	g(\mu) = (1-w)\cdot\delta_0(\mu) + w\cdot\gamma(\mu)
\]
with $\gamma(\mu)$ some convenient prior (like a Gaussian or Laplacian prior). $w$ is estimated via maximum likelihood on the marginal distribution of the observed sample.

Their method is different in that it estimates the signal sizes using the posterior \textit{median} and not the mean. They choose the median, \textit{inter alia}, for its thresholding property. The median estimator has a ``threshold zone" around zero, which depends on $w$, in which the corresponding signal size estimate is set to zero, should the observation fall there. The reader is referred to the paper for further details.

\subsection{Classical bias correction and resampling techniques}
Bias corrections for the winner's curse have been diligently sought in the biological sciences. Examples include \citet{ZollnerPritchard}, \citet{SumBull2005} and \citet{ZhongPrentice2008}. The latter also uses the truncated Gaussian distribution to produce bias corrected estimates for the largest of the estimated odds ratios in a logistic regression. The authors appeal to maximum likelihood asymptotic theory, whereas our result holds exactly for finite samples.

\citet{SimonSimon2013} use the parametric bootstrap --- with mean parameter the MLE for $\mu$: $y$ --- to estimate the bias
\[
 E[Y_{(k)} - \mu_{i(k)}]
\]
where $Y_{(1)} \geq Y_{(2)} \geq ... \geq Y_{(n)}$ and $Y_{i(k)} = Y_{(k)}$. They acknowledge that their estimate of bias is itself biased, because they generate samples with mean vector $y$ and not $\mu$. They suggest a second order bootstrap to estimate this additional bias as well, adjusting their estimates accordingly. In a simulation study in their paper, they compare the performance of their estimators (over the whole mean vector) to that of \cite{WagerDens}. Although these methods perform similarly, the empirical Bayes estimator of \cite{WagerDens} seems to have a slight edge.

\subsection{Threshold estimators}
\label{sec:thresh_esti}
Some estimators define an explicit threshold, below which the signal size estimate is set to 0. If above the threshold (absolutely), the raw sample element is either left alone (hard threhsolding) or shrunken in some way (e.g. soft thresholding). Note that the hard thresholding rule suffers from the winner's curse. Soft thresholding, with its shrinkage on non-thresholded elements toward zero, is more suited to reduced post selection bias.

These methods are differentiated mostly by their choice of threshold to apply. One option is the ``universal" threshold $\sqrt{2\log(n)}$, but this has been found poorly adaptive to the sparsity of the underlying signal vector. In an attempt to make the soft thresholding estimator more adaptive, \citet{DonJohn1995}, select the threshold $t$ in the interval $[0, \sqrt{2\log(n)}]$ such that it minimizes the Stein's Unbiased Risk Estimate (SURE) for soft thresholding:
\[
SURE^{ST}(t) = n + \sum_{i = 1}^n y_i^2 \wedge t^2 - 2\sum_{i = 1}^nI\{y_i^2 \leq t^2\}
\]

They call this their SURE estimator. Another, ostensibly more adaptive, estimator, called the \textit{adaptive SURE}, is also defined. Again, the reader is encouraged to find details in the reference.

Finally, \citet{AbramBenDonJohn}, select the threshold using the Benjamini-Hochberg ($BH(q)$) procedure and analyze the worst case risk properties, over ``sparse" parameter spaces of a hard thresholding estimator using this threshold. Their asymptotic results are impressive. However, their use of the hard thresholding estimator may lead to considerable finite sample post selection biases.\\

Our proposed method relies neither on an estimate of the marginal density of the sample elements, nor requires (potentially) computationally expensive resampling techniques, nor does it appeal to asymptotic theory in order to obtain its distributional results.  Rather, it applies recent results on post-selection inference for the Lasso to the problem of correcting selection bias in effect size estimation. Distributional results hold exactly for finite samples, regardless of the method used to select the non-zero effects. The FDR controlling threshold of \cite{AbramBenDonJohn} can be cast into our framed and used to adapt to signal sparsity. Furthermore, finite sample biases inherent in their hard thresholding estimator might be reduced. Our proposal is discussed next.


\section{Selection bias correction after conditioning on the set of proposed non-zero coefficients}
\label{sec:post_sel_theo}

In their paper, \citet{LeeSun2TaylorPostSel} propose a method for post-selection inference with the Lasso. They assume
\[
 y = X\beta + \epsilon
\]
with $X \in R^{n \times p}$ and $\epsilon \sim N (0, \sigma^2I)$. The Lasso solution for $\beta$ (for some \textit{fixed} regularisation parameter $\lambda$) satisfies:
\[
\hat{\beta}_\lambda = {\rm argmin}_\beta \left[\frac{1}{2}||y - X\beta||^2_2 + \lambda||\beta||_1\right]
\]
It is well known that the Lasso sets many elements of $\hat{\beta}_\lambda$ to zero, leading to a selection effect. Suppose we have solved this optimisation problem and have available an estimate $\hat{\beta}_\lambda$, associated with which we have the set $E \in \{1, 2, \dots, p\}$ of indices of the non-zero elements of $\hat{\beta}_\lambda$ i.e. $j \in E$ if $\hat{\beta}_\lambda \neq 0$. Furthermore, define the vector $z_E$ where $z_{E, j} = {\rm sign}(\hat{\beta}_{\lambda, j})$. The authors propose that inference (on, say, underlying effect sizes of the non-zero coefficients) proceed conditional on $(E, z_E)$.

They show how the conditioning event $(E, z_E)$ is equivalent to affine constraints on $y$, i.e. $\{E, z_E\} = \{A_{L, \lambda}y \leq b_{L, \lambda} \}$ where the forms of $A_{L, \lambda}$ and $b_{L, \lambda}$ are given explicitly in the paper. Furthermore, they show, for some $\eta \in R^n$:

\begin{equation}
F^{[\mathcal{V}^-, \mathcal{V}^+]}_{\eta^\top\mu, \sigma^2\eta^\top\eta}(\eta^\top y)|\{A_{L, \lambda}y \leq b_{L, \lambda}\} \sim {\rm Unif}(0, 1)
\label{eqTruncNormAb}
\end{equation}
where $F^{[a,b]}_{\mu, \sigma^2}$ denotes the CDF of a truncated Gaussian random variable on $[a, b]$, i.e.
\[
F^{[a,b]}_{\mu, \sigma^2} = \frac{\Phi\left(\frac{x - \mu}{\sigma}\right) - \Phi\left(\frac{a - \mu}{\sigma}\right)}{\Phi\left(\frac{b - \mu}{\sigma}\right) - \Phi\left(\frac{a - \mu}{\sigma}\right)}
\]
and $\mathcal{V}^-$, $\mathcal{V}^+$ are independent of $\eta^\top y$, with forms given explicitly in the paper. The vector $\eta$ can be chosen post-selection. In particular, we may choose $\eta = z_{r(k)}e_{r(k)}$, so that $\eta^\top y = |y|_{(k)}$.

We are interested in the orthogonal case, where $X = I$ (the $n \times n$ identity matrix), $\beta = \mu$ and $\sigma^2$ is known (assume $\sigma^2 = 1$ for ease of exposition). Now $E = \{i: |y_i| > \lambda \}$. Let $I_E$ be the $n \times |E|$ matrix obtained by selecting only the columns from $I$ corresponding to the indices in $E$. Similarly, let $I_{-E}$ contain only those columns with indices \textit{not} in $E$.

The matrix $A_{L, \lambda}$ and vector $b_{L, \lambda}$, which determine the affine constraints imposed by conditioning on the selected items, have particularly simple forms:
\begin{equation}
  A_{L, \lambda} = \left(
   \begin{array}{c}
     \frac{1}{\lambda}I^\top_{-E} \\
     -\frac{1}{\lambda}I^{\top}_{-E} \\
     -\textbf{diag}(z_E)I^\top_E
   \end{array}
  \right)
\label{eqALasso}
\end{equation} 

\begin{equation}
  b_{L, \lambda} = \left(
   \begin{array}{c}
     \textbf{1}\\
     \textbf{1}\\
     -\lambda\textbf{1}
   \end{array}
  \right)
  \label{eqBLasso}
\end{equation}
Notice that $A_{L, \lambda} \in R^{(2n - |E|) \times n}$ and $b_{L, \lambda} \in R^{(2n - |E|)}$.

Should one wish to estimate the underlying signal of those selected by the (orthogonal) Lasso procedure, one would choose $\eta = e_{i}$ with $i \in E$. Then the forms of $\mathcal{V}^-$ and $\mathcal{V}^+$ are also quite simple: if $z_i = {\rm sign}(y_i) = 1$, then $\mathcal{V}^- = \lambda$ and $\mathcal{V}^+ = \infty$ and if $z_i = -1$, then $\mathcal{V}^- = -\infty$ and $\mathcal{V}^+ = -\lambda$.

These results hold conditional on the signs. However, \citet{LeeSun2TaylorPostSel} suggest that inference (and here, perhaps point estimation) would be more effective once we marginalize over signs, ``unioning out'' the sign effect. The upshot being that $y_i$, $i \in E$, each have a truncated Gaussian distribution on $(-\infty, \lambda] \cup [\lambda, \infty)$ \textit{conditional given the selection procedure}. We use these distributions to compute point estimates and confidence intervals for the list of selected non-zero sample elements. Hopefully, this conditioning information, if used efficiently, would allow us to estimate more accurately the true underlying signal of each selected sample element and not the one tainted by selection bias.

Note, however, that these results hold for \textit{fixed} $\lambda$. In practice, we are unlikely to fix $\lambda$ beforehand and fit the Lasso (tantamount to soft thresholding), but would rather adopt another selection procedure. We consider two such procedures: selecting the $K$ largest (absolute) sample elements (for fixed $K$) and the Benjamini-Hochberg (BH) procedure with fixed FDR bound $q$. This implies an adaptive choice for the value $\lambda$, making it random. 

Fortunately, it can be shown that, conditioning on just a little more information, each of these procedures can be written in the form of affine constraints $\{Ay \leq b\}$ as above, with all the distributional results and other guarantees of \citet{LeeSun2TaylorPostSel} holding as before. This is discussed in the next section. For similar analysis of a marginal screening procedure, see \citet{LeeTaylorMargSel}.

\subsection{Selecting the K largest signals}
Let $E = \{r(k): k = 1, 2, \dots, K\}$, $G = \{r(K+1)\}$ and $H = \{1, 2, \dots, n\} \setminus (E \cup G)$. The top-$K$ selection procedure can be summarized by the inequalities:
\begin{itemize}
 \item $|y_i| > |y|_{(K+1)}$ for $i \in E$

 \item $-|y|_{(K+1)} \leq y_i \leq |y|_{(K+1)}$ for $i \in H$
\end{itemize}

These inequalities suggest that, apart from conditioning on $E$ and $z_E$ as before, we should condition further on $G$ and $z_G$ for us to be able to write them as affine constraints on $y$.

The selection event $\{E, z_E, G, z_G\} = \{A_{K}y \leq b_K\}$, where:
\begin{equation}
  A_{K} = \left(
   \begin{array}{c}
     I^\top_H - \textbf{1}z_{G}e_{r(K+1)}^\top \\
     -I^\top_{H} - \textbf{1}z_{G}e_{r(K+1)}^\top \\
     -\textbf{diag}(z_E)I^\top_E + \textbf{1}z_Ge_{r(K+1)}^\top
   \end{array}
  \right)
  \label{eqAfirstK}
\end{equation}
\begin{equation}
  b_K = \left(
   \begin{array}{c}
     \textbf{0}\\
     \textbf{0}\\
     \textbf{0}
   \end{array}
  \right)
  \label{eqBfirstK}
\end{equation}
If we let $i \in E$, then with $\eta = e_i$, we again have simple forms for $\mathcal{V}^-$ and $\mathcal{V}^+$: if $z_i = 1$, then $\mathcal{V}^- = |y|_{(K+1)}$ and $\mathcal(V)^+ = \infty$, while if $z_i = -1$, then $\mathcal{V}^- = -\infty$ and $\mathcal{V}^+ = -|y|_{(K+1)}$. We union out the signs as for the Lasso and each $y_i$, $i \in E$, has a truncated Gaussian distribution on $(-\infty, -|y|_{(K+1)}] \cup [|y|_{(K+1)}, \infty)$, conditional on the selection procedure.

\subsection{Benjamini-Hochberg selection procedure}
Let $E$, $G$ and $H$ be as before, with $J = G \cup H$. Notice now that $K$ is not fixed, but determined (randomly) by the Benjamini-Hochberg procedure, which is characterized by the inequalities:
\begin{itemize}
 \item $|y_i| - \hat{t}_{K} > 0$ for $i \in E$ where $\hat{t}_K = \Phi^{-1}\left(1 - \frac{qK}{2n}\right)$
 \item $|y|_{(k)} \leq \Phi^{-1}\left(1 - \frac{qk}{2n}\right)$ for  $k = K+1, K+2, \dots, n.$
\end{itemize}
We would need to condition on $K$ before we have any hope of extracting affine constraints on $y$. Indeed, we condition on $E$ and $z_E$ as before. In addition, we condition on the ordering of the elements in $J$, given by a permutation matrix $P_J$ such that $P_Jy = (y_{r(K+2)}, y_{r(K+3)}, \dots, y_{r(n)})^\top$. This allows us to express the selection event as an affine restriction on $y$, determined by the matrices:
\begin{equation}
  A_{BH(q)} = \left(
   \begin{array}{c}
     P_JI_J^\top \\
     -P_JI_J^\top \\
     -\textbf{diag}(z_E)I_E^\top 
   \end{array}
  \right)
  \label{eqABH}
\end{equation} 
\begin{equation}
  b_{BH(q)} = \left(
   \begin{array}{c}
     q_J \\
     q_J \\
     -\textbf{$\hat{t}_K$}
   \end{array}
  \right)
  \label{eqBBH}
\end{equation}
where $q_J = (q_{|E|+1}, q_{|E|+2}, \dots, q_n)^\top$ and $q_i = \Phi^{-1}\left(1 - \frac{qi}{2n}\right)$. Similarly, as before, each $y_i$, $i \in E$, has a truncated Gaussian distribution on $(-\infty, -\hat{t}_K] \cup [\hat{t}_K, \infty)$, conditional on the selection procedure. Fortunately, the large amount of additional information conditioned on (ordering of items in $J$) does not affect the post-selection distribution of the detected signals. If it had, we may have encountered rather noisy and unstable signal size estimates (both point and interval).

\subsection{Estimating the signal size post-selection}
Estimation of the underlying signal size after selection proceeds by maximising the (conditional) likelihood. Recall that for each of the largest-$K$ and Benjamini-Hochberg procedures, we have $y_i$, $i \in E$, distributed as truncated Gaussian over interval $(-\infty, -|y|_{(K+1)}] \cup [|y|_{(K+1)}, \infty)$, with $K$ fixed in the former, and on $(-\infty, -\hat{t}_K] \cup [\hat{t}_K, \infty)$ in the latter. The signal size estimate for the $k^{th}$ largest sample element ($k \leq K$) satisfies:

\[
 \hat{\mu}_{r(k)} = {\rm argmax}_\mu\left[\frac{\phi\left(y_{r(k)} - \mu\right)}{\Phi\left(-\hat{\lambda}^{(K)}_k - \mu\right) + 1 - \Phi\left(\hat{\lambda}^{(K)}_k - \mu\right)}\right]
\]

Equivalently (after taking logs, a single derivative and setting to zero), $\hat{\mu}_{r(k)}$ satisfies:
\begin{equation}
y_{r(k)} = \hat{\mu}_{r(k)} + \frac{\phi\left(\hat{\lambda}^{(K)}_k- \hat{\mu}_{r(k)}\right) - \phi\left(\hat{\lambda}^{(K)}_k+ \hat{\mu}_{r(k)}\right)}{\Phi\left(-\hat{\lambda}^{(K)}_k - \hat{\mu}_{r(k)}\right) + 1 - \Phi\left(\hat{\lambda}^{(K)}_k- \hat{\mu}_{r(k)}\right)}
\label{eqCondMLE}
\end{equation}
where $\hat{\lambda}^{(K)}_k = |y|_{(K+1)}$ for the largest-$K$ and $\hat{\lambda}^{(K)}_k = \hat{t}_K$ for the $BH(q)$ procedure.

Equation~\eqref{eqCondMLE} defines our proposed estimator for signal size after selection. As an illustration of the behaviour of this effect size correction, consider a simpler form of \eqref{eqCondMLE}:

\begin{equation}
y = \mu + \frac{\phi(\lambda - \mu) - \phi(\lambda + \mu)}{\Phi(-\lambda - \mu) + 1 - \Phi(\lambda - \mu)}
\label{eqCondMLESimple}
\end{equation}

\begin{figure}[htb]
  \centering
  \includegraphics[width=90mm]{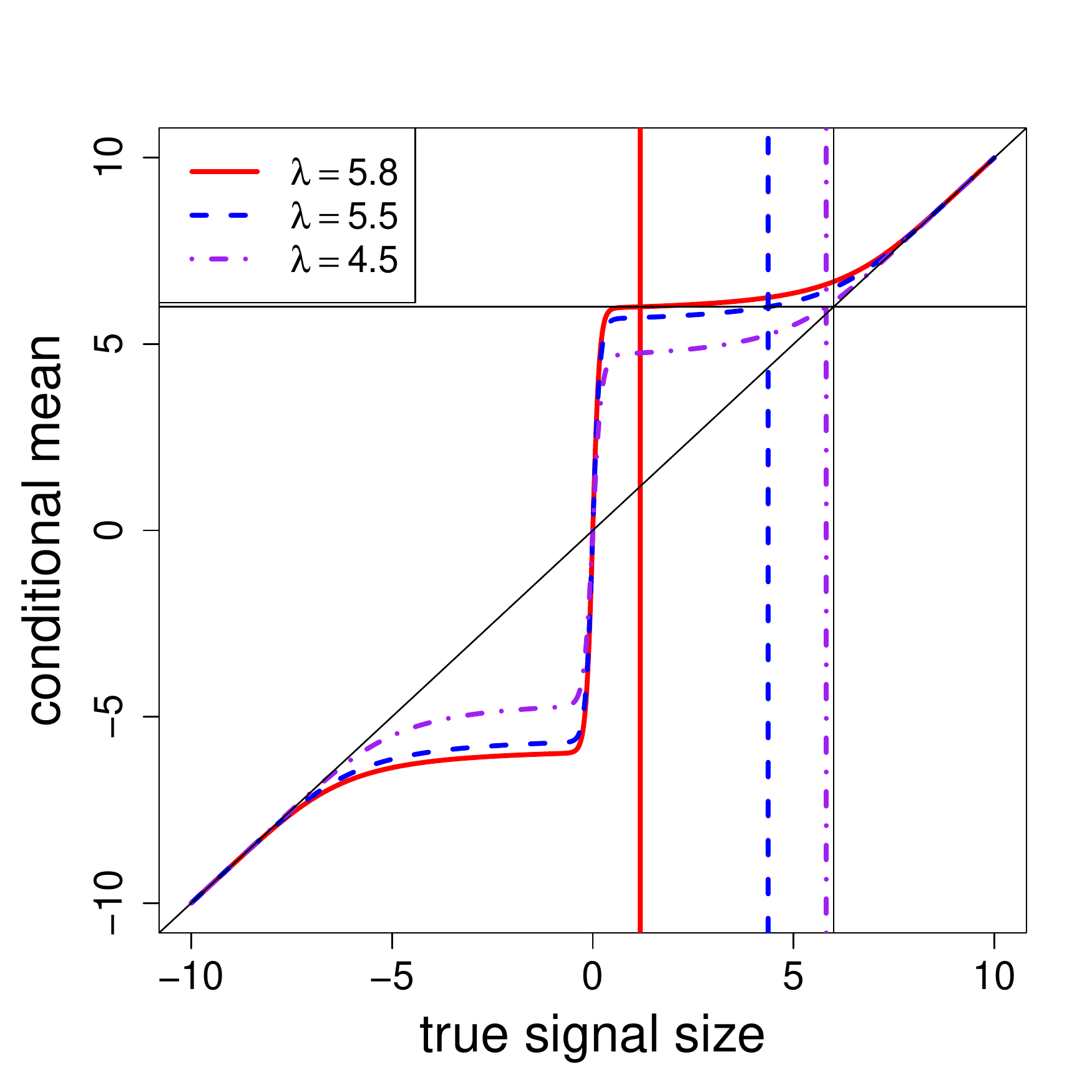}
  \caption{\emph{Plot of conditional mean of a Gaussian random variable truncated to $(-\infty, \lambda] \cup [\lambda, \infty)$. Horizontal and vertical black lines correspond to $y = 6$ and represent the observed order statistic and maximum (unconditional) likelihood estimate for the signal size. Solid red curve represents conditional expectation in \eqref{eqCondMLESimple} as function of $\mu$ for $\lambda = 5.8$. Blue and purple dashed lines represent conditional expectations for $\lambda = 5.5$ and $\lambda = 4.5$ respectively. Vertical lines represent maximum (conditional) likelihood signal size estimates for corresponding coloured expectation curves. These are obtained by reading off the horizontal coordinate of intersection between the horizontal black line and relevant expectation curve.}}
  \label{figCondMean}
\end{figure}

Figure~\ref{figCondMean} plots the expectation in \eqref{eqCondMLESimple} as a function of the true signal size $\mu$ for different values of $\lambda$, namely $\lambda = 4.5, 5.5, 5.8$ and $y = 6$. Notice how the maximum conditional likelihood estimate of signal size shrinks the observed $y$ toward zero, with the extent of the shrinkage determined by how close the threshold $\lambda$ is to $y$. The closer to $y$ the $\lambda$; the more the shrinkage. Notice that this shrinkage is not the same as the scaling, translation or keep-or-kill type shrinkages we get from ridge regression, Lasso and $\ell_0$ penalties, respectively. Successful signal size estimates obviously depend greatly on the chosen threshold $\lambda$. Simulation results presented in the next section postulate a seemingly good method for the selection of this threshold. 

Here we have considered the orthogonal case, whereby $Y$ is assumed to have a diagonal covariance matrix. One can extend the theory and obtain estimators for the correlated case as well. The discussion being currently tangential to our interest, we defer this discussion to future work.


\section{A simulation study}\label{sec:simulation}
We ran a simulation study comparing the performance of our truncated Gaussian shrinkage estimator against that of some other candidate estimators. Samples of size $n = 1000$ were generated independently from a $N(\mu, I)$ distribution. Sparsity of the mean vector was controlled via the parameter $\alpha$ so that $\lceil n^\alpha \rceil$ entries were non-zero (the rest: zero). Non-zero entries were generated independently from a univariate $N(\nu, 1)$ distribution. 

Each sample was sorted in decreasing order of (absolute) size of the sample elements. Signal size estimates were computed for each using a variety of methods. Methods were compared via their mean squared error (MSE) over the first $K$ elements \textit{only}:
\[
MSE(K) = \frac{1}{K}\sum_{k = 1}^K\left(\hat{\mu}_{r(k)} - \mu_{r(k)}\right)^2
\]
We consider MSE over only a segment of the underlying signal vector, because we are interested in comparing the performance of the estimators post selection. We know the average MSE for the unconditional likelihood estimator \textit{over the entire vector} is 1 and that it is dominated by the estimator of \citet{JamesStein}. However, after selection, considering only a small segment of the total vector, we have little theory to guide us. Hence the rationale for the simulation. It is known that the hard thresholding estimator does very poorly over the largest of the signals (see Section~\ref{sec:introduction}) and the hope is that some other methods perform better. 

The list of methods considered here comprises of the following:

\begin{enumerate}
	\item Unconditional maximum likelihood estimator: $\hat{\mu}_{r(k)} = y_{r(k)}$. Note that this is equivalent to a hard thresholding (HT) estimator. If we choose $K$ via the $BH(q)$ procedure, this becomes the estimator analysed by \citet{AbramBenDonJohn}.
	\item Soft thresholding (ST) estimator with threshold chosen as $|y|_{(K+1)}$: $\hat{\mu}_{r(k)} = {\rm sign}(y_{r(k)})(|y_{r(k)})| - |y|_{(K+1)})$.
	\item James-Stein shrinkage estimator (JS), shrinking toward zero: $\hat{\mu}^{JS} = \left(1 - \frac{n-2}{||y||^2}\right)y$ and $\hat{\mu}_{r(k)} = \hat{\mu}^{JS}_{r(k)}$.
	\item Truncated Gaussian (TN) shrinkage estimator of Equation~\eqref{eqCondMLE}.
	\item Second-order bootstrap estimator of \citet{SimonSimon2013} (B2).
	\item Oracle estimator of \citet{SimonSimon2013} (B-O). The oracle estimator cannot be used in practice, because it generates parametric bootstrap samples using the true underlying signal vector $\mu$. It is privy to much more useful information than are the other estimators and is a good yardstick against which to measure the relative performance of the other estimators.
	\item Empirical Bayes estimator of \citet{EfronTweedie} where the density is estimated using the non-linear programming technique of \citet{WagerDens} (NLP).
	\item Empirical Bayes thresholding estimator of \citet{JohnSilver2004} (EBT). We use the default settings suggested by the authors. Note that the threshold chosen by this method is not the same as the threshold implicitly suggested by only considering the top $K$ sample elements. We still consider the top $K$ sample elements, even if this method produces some zero mean estimates.
	\item Empirical Bayes estimator with gold standard, generalized maximum likelihood estimator of the density as in \citet{JiangZhangMLDens} (GMLEB).
		
\end{enumerate}

Two types of selection procedures were considered. The first selected the largest $K$ sample elements and deemed them to have non-zero signals. Here $K$ was deterministic and was varied. The second procedure was the Benjamini-Hochberg (BH) procedure of \citet{BH95}, with false discovery rate (FDR) parameter, $q$, varied. The estimators are numbered exactly as in this list in what follows. Notice that estimators 3 and 5-9 provide estimates for the entire signal vector. One cannot apply these methods to a fraction of the vector only. We applied the method, obtaining an estimate of the entire mean vector, but then considered the MSE over only those sample elements selected by the top-$K$ or $BH(q)$, whichever was appropriate.

Some additional estimators were also considered, but not included in the output. One is the first order bootstrap estimator of \citet{SimonSimon2013}. Its performance was always similar to, but slightly poorer than that of the second order bootstrap estimator. Other estimators had very poor MSE performance and, since we cut off  the plots at some maximum MSE, their output did not make it into the final presentation. These include the soft thresholding estimator at the universal threshold, as well as the SURE and adaptive SURE estimators discussed in Section~\ref{sec:thresh_esti}.

Simulation parameters that were varied include:
\begin{itemize}
	\item	Sparsity parameter: $\alpha = 0.1, 0.15, \dots, 0.5$
	\item	Signal strength parameter: $\nu = 3, 4, 5, 6$
	\item Number of sample elements selected in top-$K$ selection procedure: $K = 1, 2, \dots, 30$.
	\item False discovery rate control parameter in the $BH(q)$ procudure: $q = 0.05, 0.1, 0.15, 0.2, 0.3, 0.5$.
\end{itemize}

$S = 55$ simulations were run at each setting of the parameters. For each setting, each of the signal size estimates was computed and its (partial) MSE computed. Some results are presented below.

\subsection{Top-$K$ selection: varying the number selected ($K$)}
In practice, we would not know how many true non-zero signals there are underlying our sample elements. Some of the methods mentioned (e.g. bootstrap of \citet{SimonSimon2013} and the empirical Bayes of \citet{WagerDens}) estimate the entire signal vector, with the user subsequently allowed to select $K$ and consider only those elements in the estimated vector. Others, including the truncated Gaussian shrinkage estimator, only estimate the signals of the \textit{selected} sample elements. Here an appropriate choice of $K$ becomes important. Clearly, the optimal $K$ selected must bear some relation to the sparsity of the underlying signal vector, here controlled by the $\alpha$ parameter . The \textit{smaller} this parameter; the \textit{more} sparse is $\mu$.

\begin{figure}[htb]
	\centering
	\includegraphics[width=120mm]{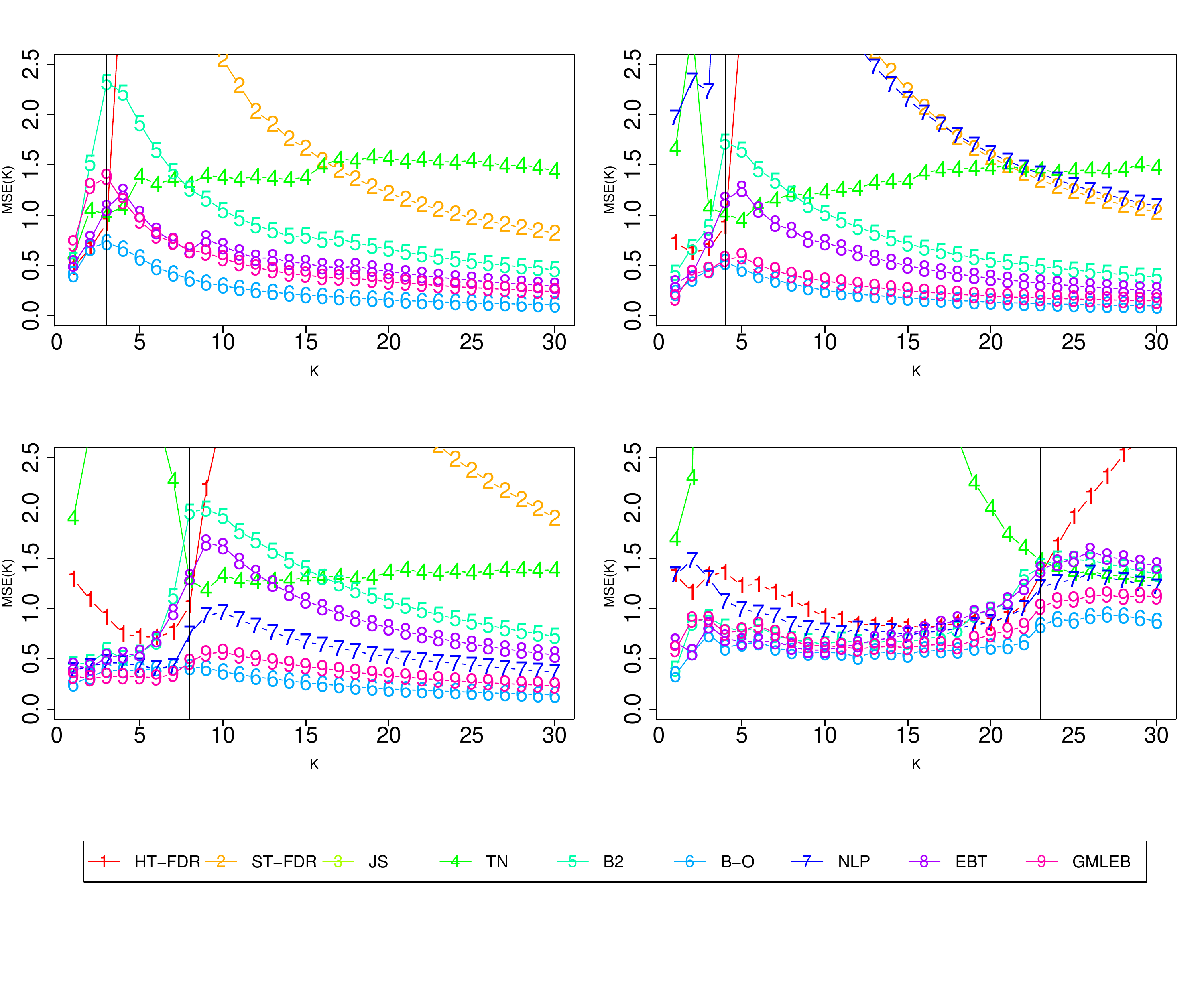}
	\caption{\emph{Median $MSE(K)$ as function of number of deterministically selected sample elements ($K$). Estimators numbered as in the numbered list. Sparsity varies over panels. Top left: $\alpha = 0.15$, top right: $\alpha = .2$, bottom left: $\alpha = 0.3$ and bottom right: $\alpha = 0.45$. Vertical black lines at $\lceil n^\alpha \rceil$ --- the true number of non-zero signals. Figure is truncated above at $MSE(K)=2.5$. Sample size $n = 1000$; signal size $\nu = 6$. Our estimator of Equation~\eqref{eqCondMLE} is labelled `TN' and is in green (number 4). Curve labelled number 6 is the oracle and is not computable in practice.}}
	\label{figVaryK}
\end{figure}

Figure~\ref{figVaryK} plots the median (over the $S = 55$ simulations) partial MSE ($MSE(K)$) as a function of $K$. Each panel shows the curves for a different sparsity setting. Vertical black lines are drawn at $\lceil n^\alpha \rceil$ --- the true number of non-zero signals. 

Hard thresholding and James-Stein estimators (numbered 1 and 3) seem to perform quite poorly, their curves often running off the top of the plot, which is bounded above to facilitate the viewing of the curves of the other methods. Clearly, we do better with other methods.
 
The oracle bootstrap estimator outperforms all others at all values of $K$. This is unsurprising considering the extra information it is privy to. The median MSE curves of this estimator appears \textit{always} to lie below those of the others (even in subsequent figures). One cannot apply this estimator in practice, because it relies on the true mean in the parametric bootstrap step -- the very thing we are trying to estimate. This estimator should not be viewed as a competitor we are looking to beat. It is merely a reference.

Empirical Bayes estimates (EBT and GMLEB) perform well at all sparsity levels. The NLP empirical Bayes estimator performs very poorly in the presence of very sparse signals, but sees a dramatic improvement relative to other methods as the underlying signal becomes less sparse.

It is interesting to consider the behavior of the truncated Gaussian estimator. When $K < \lceil n^\alpha \rceil$ (left of vertical black lines), the method generally does poorly. We are still in the signal variables here and the method obviously does not do well when we use one of the non-zero signal elements as a threshold. Performance improves markedly as $K$ approaches and goes beyond the true number of non-zero signals, after which it competes with the best of the methods over a significant range of $K$. In particular, it would seem that the method slightly outperforms the other methods for values of $K$ just to the right of the true number of non-zero signals, especially for the sparse signal vectors. This phenomenon seems to hold for slightly weaker signals as well. 

One should bear in mind that a good estimator here does not necessarily have to outperform all others at all values of $K$. In practice, we would use some method to select $K$, either by rule-of-thumb or in some more principled way (like that of Benjamini and Hochberg, for example). All we require is for our estimator to be good \textit{over the selected sample elements}, which corresponds to a single $K$ in each of these plots. 

\subsection{Effect of sparsity ($\alpha$) at the true number of signals ($K^\star (\alpha)$)}

As demonstration of this point, suppose we were told the value of $\alpha$. Then we could compute the true number of non-zero signals $K^\star(\alpha) = \lceil n^\alpha \rceil$ and compare the performance of the different estimators at this $K$, varying the sparsity parameter $\alpha$. 

\begin{figure}[htb]
	\centering
	\includegraphics[width = 120mm]{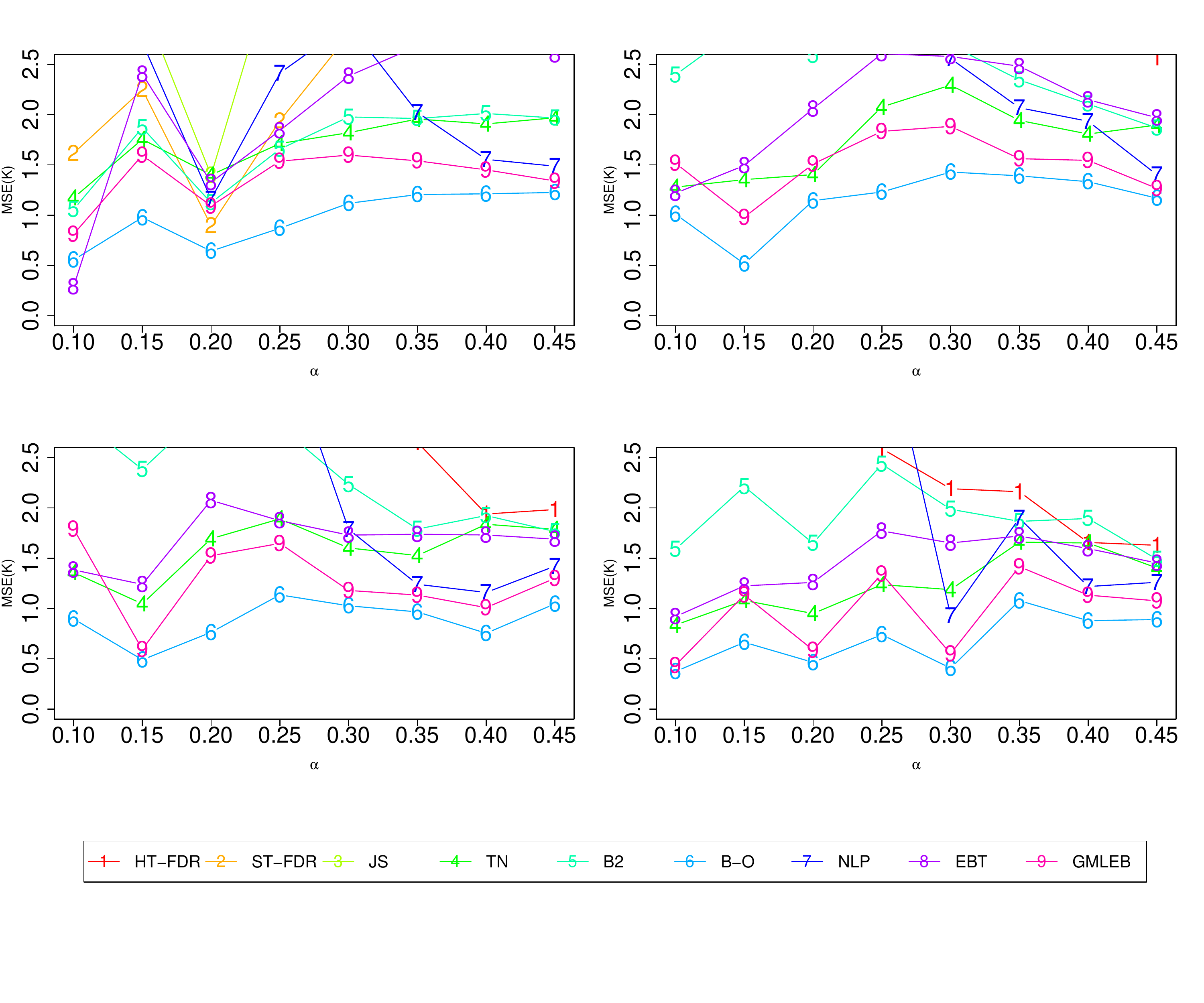}
	\caption{\emph{Median $MSE(K)$ as a function of sparsity parameter ($\alpha$) at the true number of signals ($K^\star (\alpha)$)  plus 1. Curves numbered as in the numbered list. Signal size ($\nu$) varies over panels, increasing from 3 in the top left panel, to 6 in the bottom right. Figure truncated above at $MSE(K) = 2.5$. Sample size $n = 1000$. Our estimator of Equation~\eqref{eqCondMLE} is labelled `TN' and is in green (number 4). Curve labelled number 6 is the oracle and is not computable in practice.}}
	\label{figVaryAlphaTrueK}
\end{figure}

Figure~\ref{figVaryAlphaTrueK} plots the median partial MSE as a function of the sparsity parameter at a $K$ just to the right of $K^\star (\alpha)$. Our truncated Gaussian estimator seems to outperform most others over a fairly broad range of sparsity values at a handful of signal strength settings, competing with (and sometimes improving on) the gold standard GMLEB estimator. It is only once signals become less sparse that the other estimators (like NLP, EBT and HT) catch up. Even then the truncated Gaussian estimator does not perform much more poorly than the best.

Of course, in practice, we are never told the true value of $\alpha$ and we cannot compute $K^\star (\alpha)$. We rely on other methods for determining the $K$ used in further analysis. One such method is the procedure of Benjamini and Hochberg at false discovery rate control level $q$ ($BH(q)$). The hope is that the performance of the truncated Gaussian estimator is still reasonable without the oracle-like knowledge of exactly how many non-zero signals there are. Results from signal size estimation post such selection are detailed next.

\subsection{Selection via Benjamini-Hochberg$(q)$: varying $q$}

\begin{figure}[htb]
	\centering
	\includegraphics[width=130mm]{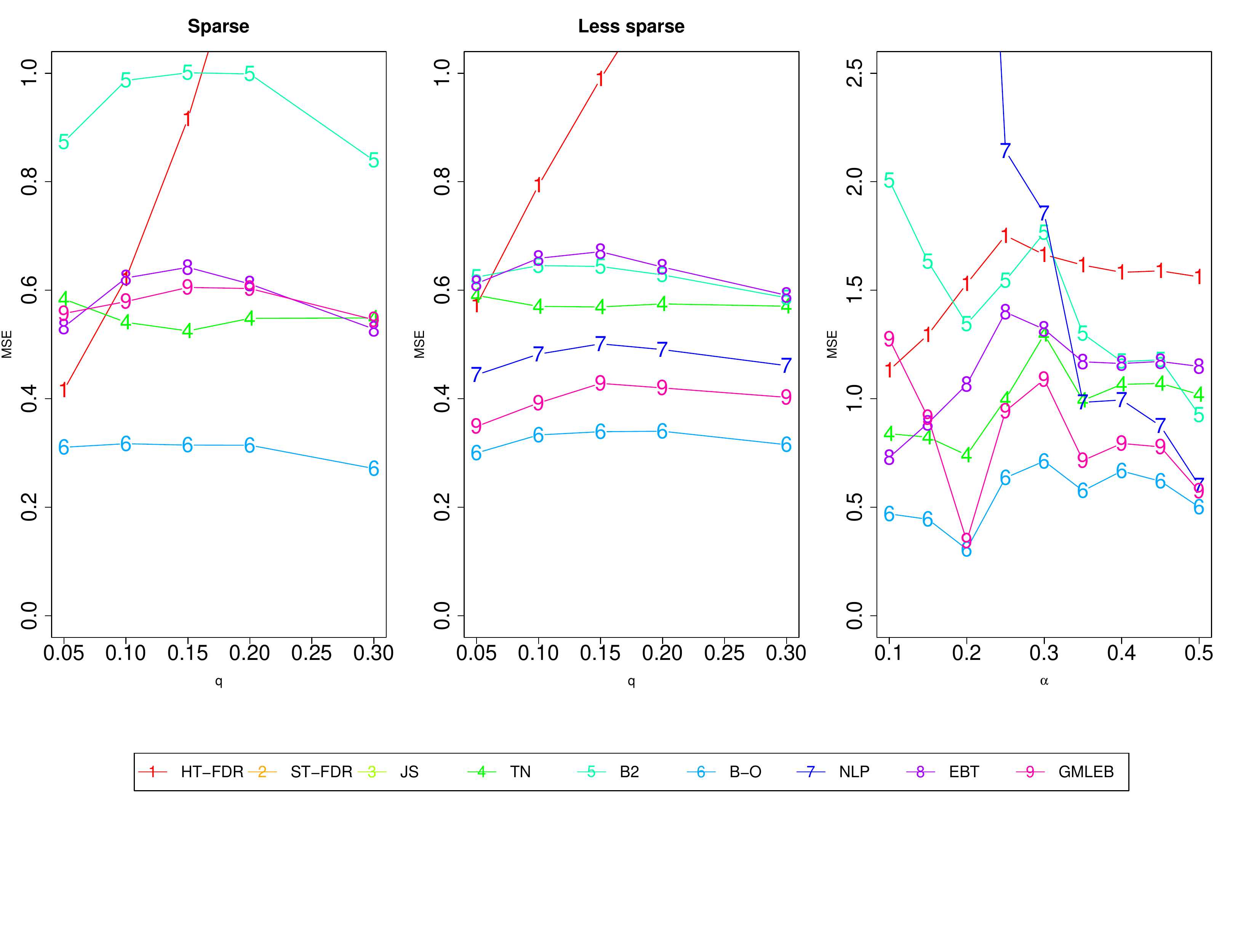}
	\caption{\emph{``Integrated median partial MSE". Curves labeled ``6" are the oracle bootstrap estimators of Simon and Simon (2013). These cannot be computed in practice. Our truncated Gaussian estimator is labelled ``4". It performs best in the sparsest settings. \textbf{Left panel}: Median partial MSE curves as function of $q$ (as in Figure~\ref{figBHVaryQ}; each obtained over $S = 55$ replications) are pointwise-averaged over $\alpha = 0.1, 0.15, 0.25$ -- the sparsest signal vectors -- with the resulting curves plotted as function of $q$. \textbf{Middle panel}: As in the left panel, but with averaging over $\alpha = 0.3, 0.35, \dots, 0.5$ -- less sparse signal vectors. \textbf{Right panel}: Median partial MSE curves as function of $\alpha$ (as in Figure~\ref{figBHVaryAlpha}; each obtained over $S = 55$ replications) are pointwise-averaged over $q = 0.05, 0.1, 0.15, 0.2, 0.3$, with the resulting curves plotted as function of $\alpha$.}}
	\label{fig:integratedMSE}
\end{figure}

The conservatism of the Benjamini-Hochberg procedure is controlled by the false discovery rate limit parameter ($q$). We are guaranteed that the expected ratio of false discoveries (zero signals selected) is bounded above by $q$. This conservatism comes at the cost of not selecting all signal variables. Although we do not have direct control over the number of sample elements as we did in the previous paragraphs (by controlling $K$), we have indirect control via $q$. The higher the $q$; the less conservative the Benjamini-Hochberg procedure and the more sample elements are selected as non-zero each time. 

\begin{figure}[htb]
	\centering
	\includegraphics[width=120mm]{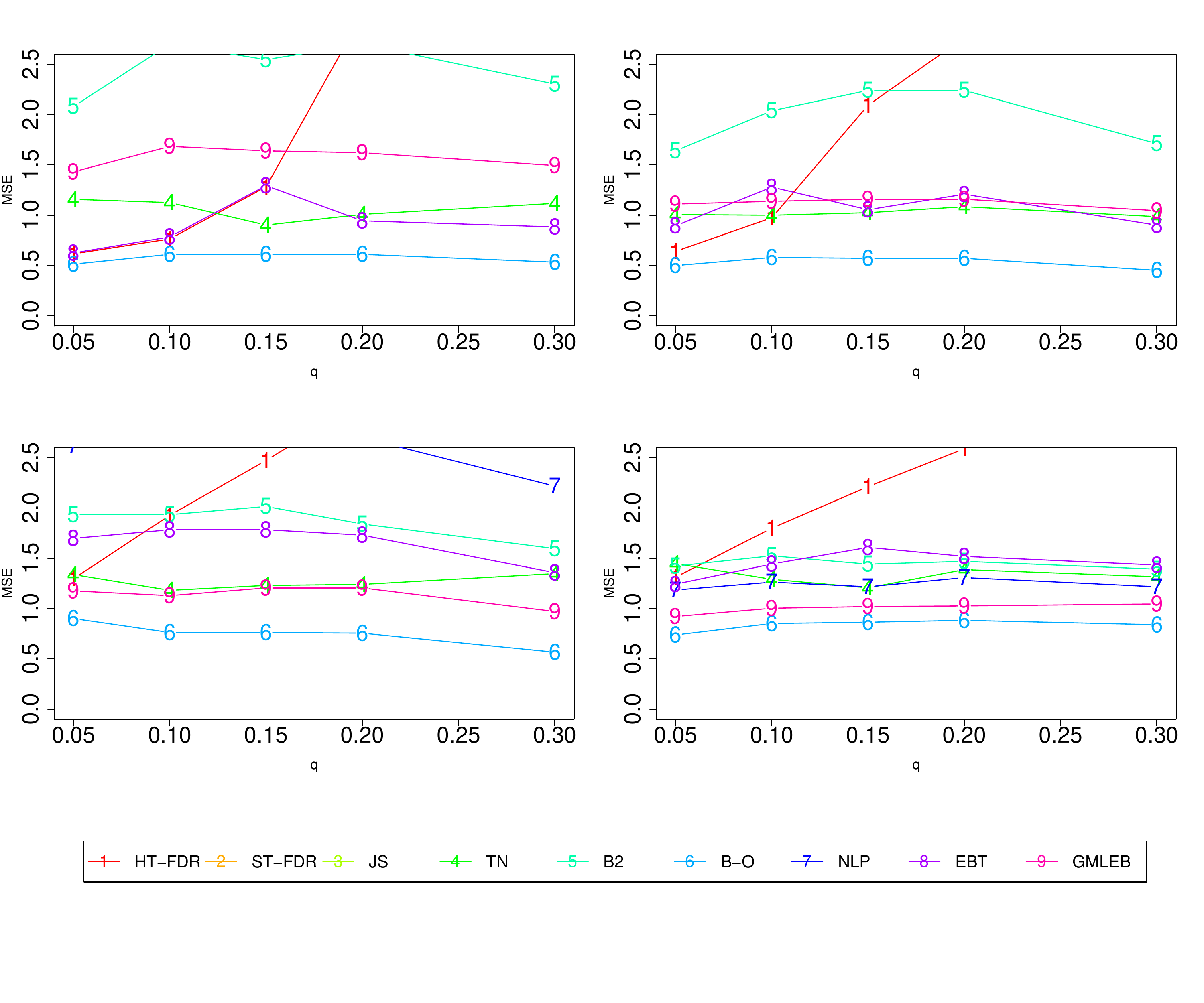}
	\caption{\emph{Median $MSE(K)$ as function of the FDR limit ($q$) of the $BH(q)$ procedure. Curves labelled as in numbered list. Sparsity of signal vector varies over panels. Top left: $\alpha = 0.1$, top right: $\alpha = 0.15$, bottom left: $\alpha = 0.25$ and bottom right: $\alpha = 0.4$. Figure truncated above at $MSE(K)=2.5$. Sample size $n = 1000$, signal size $\nu = 6$. Our estimator of Equation~\eqref{eqCondMLE} is labelled `TN' and is in green (number 4). Curve labelled number 6 is the oracle and is not computable in practice.}}
	\label{figBHVaryQ}
\end{figure}

Figure~\ref{fig:integratedMSE} plots ``integrated median partial MSE" for the different estimators. This figure serves as a summary to the subsequent two figures in this section. Details are given in the caption and later in the section. The upshot is that our truncated Gaussian estimator performs best (and outperforms all contenders, except the oracle, which is not computable in practice) for sparse signals (left panel). The best performance obtains at $q = 0.1, 0.15, 0.2$: FDR control settings we believe most widely used in practice. As signals become less sparse, other estimators perform better, but ours remains in the top three best performers (middle and center panels). Furthermore, the methods producing estimators 7 and 9 are not easily applied to the computation of confidence intervals. Contrast this with our method, which provides a fully parameterized post-selection distribution of the selected sample elements, leading quite simply and naturally to the construction of valid post-selection confidence intervals. We pursue this further in Section~\ref{sec:ci}.

Figure~\ref{figBHVaryQ} plots the median partial MSE as a function of $q$ for different sparsity levels. Notice how the truncated Gaussian estimator performs admirably, competing with  the best estimator at levels $q = 0.1$, $q = 0.15$ and $q = 0.25$. Although it is never the outright best (except maybe at a few points), the other competitors' performance vary according to the sparsity settings. For example, the B2 and NLP estimators do poorly at sparse signals and improve as signals become less sparse, although they never outperform our TN estimator. The GMLEB estimator improves as sparsity decreases, settling into pole position when $\alpha = 0.45$. Finally, the EBT estimator competes with our TN estimator at high sparsity levels, with the latter improving relatively as sparsity decreases. This is heartening, since these are three settings for $q$ that are likely to be used in practice. It seems to perform best at moderate sparsity levels.

The NLP empirical Bayes estimator performs poorly in the sparsest cases, probably because it requires moderately dense signals to produce decent density estimates. Its performance improves markedly as sparsity decreases. Hard thresholding, James-Stein and soft thresholding perform very poorly at all levels of sparsity.

Of course, in practice we would not try many different settings of $q$, settling rather on one and hoping that our signal size estimator produces decent estimates \textit{at all or many} levels of sparsity. From Figure~\ref{figBHVaryQ} it seems as though the levels $q = 0.1$, $q = 0.15$ and $q = 0.2$ could be promising levels.

\begin{figure}[htb]
	\centering
	\includegraphics[width=120mm]{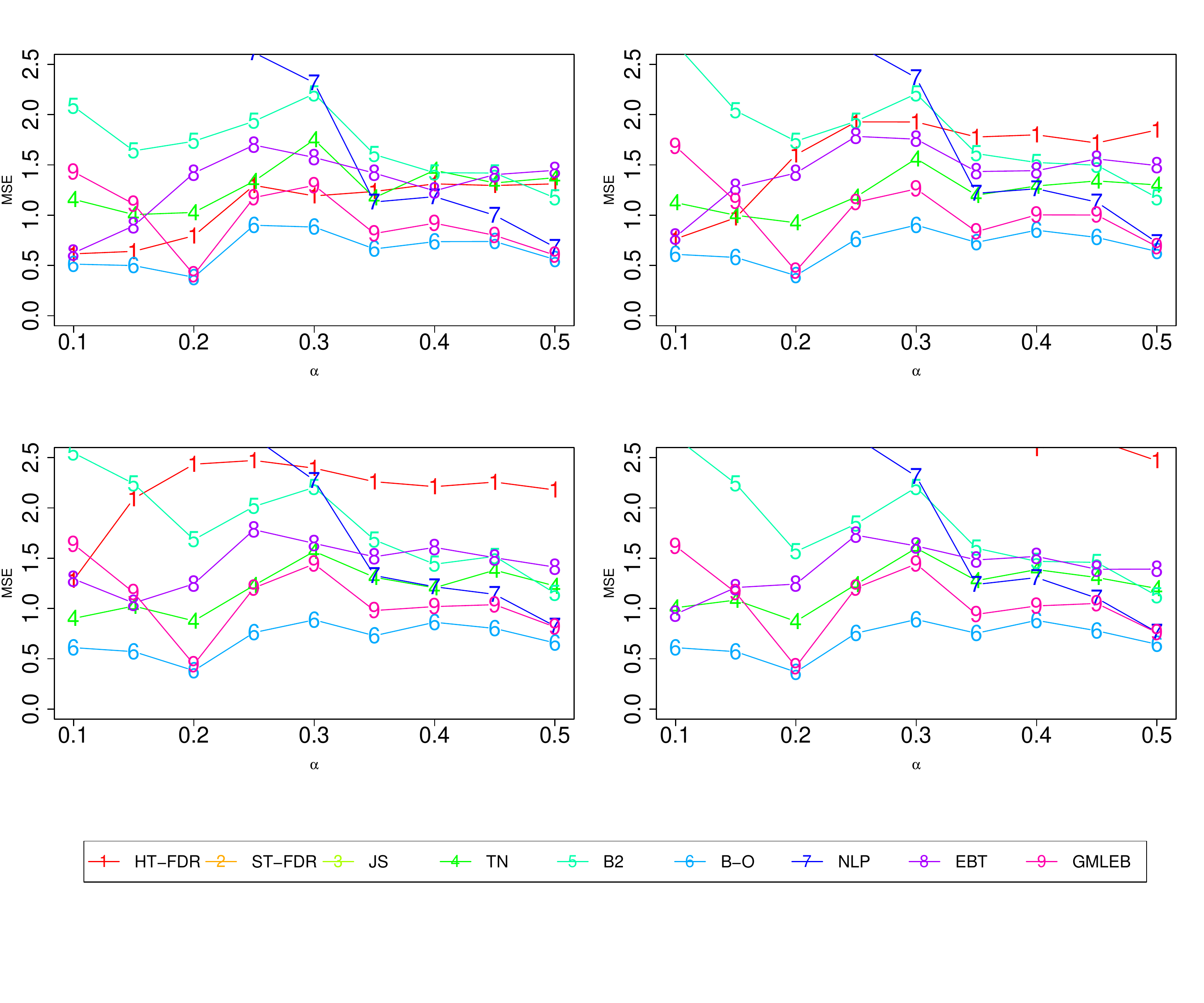}
	\caption{\emph{Median $MSE(K)$ as function of sparsity level for the $BH(q)$ selection procedure. The value of $q$ varies over the panels. Top left: $q = 0.05$, top right: $q = 0.1$, bottom left: $q = 0.15$ and bottom right: $q = 0.2$. Figure truncated above at $MSE(K) = 3$. Sample size $n = 1000$, fixed signal size $\nu = 6$. Our estimator of Equation~\eqref{eqCondMLE} is labelled `TN' and is in green. Curve labelled number 6 is the oracle and is not computable in practice.}}
	\label{figBHVaryAlpha}
\end{figure}

Figure~\ref{figBHVaryAlpha} plots median partial MSE against sparsity at different levels of $q$. The truncated Gaussian estimator competes at most levels of sparsity in just about every panel. In particular, the panels corresponding to $q = 0.1$, $q = 0.15$ and $q = 0.2$ suggest that this procedure is in the top two or three (non-oracle) performers over most sparsity levels. The truncated Gaussian post-selection signal size estimator seems to perform well in conjunction with the Benjamini-Hochberg procedure for $q = 0.1$, $q = 0.15$ and $q = 0.2$. 

The good performance of the hard thresholding estimator (red curves, labelled 1) in parts of Figure~\ref{figBHVaryQ} and the top panels of Figure~\ref{figBHVaryAlpha} may initially seem contradictory, especially considering the evidence of Figure~\ref{figWinnersCurse}. Does the winner's curse not dictate that this estimator perform worst with respect to partial MSE over the first few signals? The apparent conundrum can be resolved when one considers those sample elements with large underlying signals as a separate subsample from those without. The signal elements form a small sample --- one in which the winner's curse (as evinced by large positive errors) is less prominent. For the handful of signal elements, the (absolutely) sorted vector of sample elements may actually serve as decent estimates to the underlying signals. 

For sparse signals then, we have a small subsample of signal elements and, for low $q$, the conservatism inherent in the FDR control guarantee of the Benjamini-Hochberg procedure ensures that only some of these signal elements are selected (and few or none of the zero-signal elements). Computing partial MSE only over these few sample elements leads to fairly small error values, as seen in the figures. As we increase $q$, however, our Benjamini-Hochberg procedure becomes less conservative and it is allowed to select some true zero-signal elements as well. It tends to select the largest of these sample elements, leading to a rapidly increasing partial MSE as they are included in the calculation. Hence the shape of the curves in Figure~\ref{figBHVaryQ}. A combination of signal sparsity and small $q$ tend to work together to keep partial MSE contained. Notice however, that the post-selection truncated Gaussian estimator does not suffer as we increase $q$ or decrease sparsity. It manages to defuse the explosive effect of the first few zero-signal sample elements selected.

Since the estimators seem to exhibit varying performance over the different parameter settings, it might be revealing to amalgamate the information in the last two plots. To this end, we define an ``integrated median MSE", which is merely the pointwise, arithmetic mean of the median partial MSE curves, a sample of which have been shown in the panels of the preceding two plots. We expect the truncated Gaussian estimator to perform reasonably well here, because of its consistent performance over a range of parameter settings. This is shown in Figure~\ref{fig:integratedMSE}. 

Indeed, it seems to be the case that the truncated Gaussian estimator performs admirably, settling into first place amongst those estimators actually computable in practice (for sparsest signals) and third for less sparse signals. These plots provide a good summary of the most notable performance trends: the TN, EBT and GMLEB estimators perform consistently over all panels; hard thresholding performs poorly when the FDR control parameter is increased and performs more poorly overall than TN, EBT and GMLEB; B2 and NLP require signals to be less sparse before becoming effective.

Overall, it would seem that our truncated Gaussian estimator performs admirably. The reader should not be misled by the figures in this section and discount its performance. We have shown only the best performing estimators, truncating our figure panels to focus exclusively on the area inhabited by these good estimators. Even though it does not systematically outperform \textit{all} competitors, it is the most consistent of the bunch, giving reasonably accurate estimates of the selected signals at large to moderate sparsity levels and FDR parameter values likely to be used most in practice.

Recall also that out method lends itself naturally and seamlessly to the construction of valid post-selection confidence intervals. We discuss this further in Section~\ref{sec:ci}. This property is not as simply embodied in, for example, the procedure leading to estimator GMLEB.

\section{An asymptotic upper bound on worst case risk on sparsity balls}\label{sec:risk}
In the previous section we saw that the truncated Gaussian estimator with FDR-selected threshold performed admirably, always competing within the top three estimators in the sparse setups encountered in the simulation. The particular recipe -- first selecting non-zero signals via $BH(q)$ and then estimating their signal sizes -- recalls the estimators analysed in \citet{AbramBenDonJohn} and \citet{JiangZhangFDRST}. It is natural to ask whether, and to what extent, the asymptotic results encountered there, pertaining to asymptotic minimax risk on parameter spaces meant to evoke sparsity, transfer to our estimator. It is very similar in spirit, after all.

Suppose we have the setup as in (\ref{eqManyMeansModel}), with $\sigma^2 = 1$, having observed data $y_1, y_2, \dots, y_n$. Defining notation, let $\hat{\mu}_{t, i}^{HT} = \hat{\mu}_{t}^{HT}(y_i) = y_iI\{|y_i| > t\}$ be the hard thresholding estimator for $\mu_i$ based on threshold $t$. Similarly, let $\hat{\mu}^{ST}_{t, i} = \hat{\mu}_{t}^{ST}(y_i) = {\rm sign}(y_i)\max\{|y_i| - t, 0\}$ and $\hat{\mu}^{TN}_{t, i} = \hat{\mu}_{t}^{TN}(y_i) = \tilde{\mu}_iI\{|y_i| > t\}$ be the soft thresholding and truncated Gaussian estimators for $\mu_i$ using the same threshold $t$. Here $\tilde{\mu}_i$ is obtained by solving (\ref{eqCondMLESimple}) for $\mu$ after setting $y = y_i$ and $\lambda = t$.

Furthermore, fix FDR parameter $q \in (0, 1)$ and let $t_k = \tilde{\Phi}^{-1}(q/2\cdot k/n)$ be the right tailed standard Gaussian quantiles (here $\tilde{\Phi}(\cdot) = 1 - \Phi(\cdot)$, with $\Phi(\cdot)$ the standard Gaussian cumulative distribution function). The FDR controlling threshold separating nulls from non-nulls is $t_{\hat{k}_F} = \hat{t}_F$ where $\hat{k}_F$ is the largest index $k$ for which $|y|_{(k)} \geq t_k$, as in the $BH(q)$ procedure.

\citet{AbramBenDonJohn} and \citet{JiangZhangFDRST} consider the worst case risk of an estimator $\hat{\mu}$ over a specified parameter space $\Theta_n$:
\begin{equation}\label{eq:worst_case_risk}
 \tilde{\rho}(\hat{\mu}, \Theta_n) = \sup_{\Theta_n} E_\mu ||\hat{\mu} - \mu||^r_r,
 \end{equation}
 $0 < r \leq 2$. In particular, they study the behaviour of the estimators $\hat{\mu}^{HT}_{\hat{t}_F}$ and $\hat{\mu}^{ST}_{\hat{t}_F}$ comprising of components $\hat{\mu}^{HT}_{\hat{t}_F, i}$ and $\hat{\mu}^{ST}_{\hat{t}_F, i}$ respectively over three types of parameter spaces, each meant to evoke sparsity in the mean vector. Spaces considered are:
\begin{enumerate}
		\item $\Theta_n = \ell_0[\eta_n] = \{\mu : ||\mu||_0 \leq \eta_n \cdot n\}$ --- the ``nearly black" case, where we limit the proportion of non-zero signals ($\eta_n$). Note how this proportion can depend on $n$.
		\item $\Theta_n = m_p[\eta_n] = \{\mu : |\mu|_{(k)} \leq C\cdot \eta_n \cdot n^{1/p} \cdot k^{-1/p}\}$ --- the ``weak $\ell_p$ ball".
		\item $\Theta_n = \ell_p[\eta_n] = \{\mu : \frac{1}{n}\sum_{i = 1}^n |\mu_i|^p \leq \eta_n^p\}$ --- the ``strong $\ell_p$ ball".
\end{enumerate}

The worst case risk $\tilde{\rho}(\hat{\mu}^{HT}_{\hat{t}_F}, \Theta_n)$ is compared to the minimax risk
\[
R_n(\Theta_n) = \inf_{\hat{\mu}}\,\tilde{\rho}(\hat{\mu}, \Theta_n)
\]
with an asymptotic upper bound for the former written in terms of the latter. We wish to adapt these results to our estimator $\hat{\mu}^{TN}_{\hat{t}_F}$. Note that this is not exactly the estimator considered in the simulation study. That was $\hat{\mu}^{TN}_{|y|_{\hat{k}_F}}$. However, the former is more amenable to analysis and is close to the latter.

We note that, for $0 < r \leq 2$:
\begin{align*}
	||\hat{\mu}^{TN}_{\hat{t}_F} - \mu||^r_r &= ||\hat{\mu}^{TN}_{\hat{t}_F} - \hat{\mu}^{HT}_{\hat{t}_F} + \hat{\mu}^{HT}_{\hat{t}_F} - \mu||^r_r \\
	&\leq 2^{(r-1)_+}\left[||\hat{\mu}^{TN}_{\hat{t}_F} - \hat{\mu}^{HT}_{\hat{t}_F}||^r_r + ||\hat{\mu}^{HT}_{\hat{t}_F} - \mu||^r_r \right] \\
	&\leq 2^{(r-1)_+}\left[||\hat{\mu}^{ST}_{\hat{t}_F} - \hat{\mu}^{HT}_{\hat{t}_F}||^r_r + ||\hat{\mu}^{HT}_{\hat{t}_F} - \mu||^r_r \right] \\
	&= 2^{(r-1)_+}\left[\hat{k}_F\hat{t}_F^r + ||\hat{\mu}^{HT}_{\hat{t}_F} - \mu||^r_r \right]
\end{align*}
where the second inequality follows from Lemma~\ref{lemma:squeeze}. Hence
\begin{equation}\label{eq:risk_bound}
	E||\hat{\mu}^{TN}_{\hat{t}_F} - \mu||^r_r \leq 2^{(r-1)_+}\left[E[\hat{k}_F\hat{t}_F^r] + E||\hat{\mu}^{HT}_{\hat{t}_F} - \mu||^r_r \right]
\end{equation}
This puts the quantity in an amenable form. Furthermore, since, also from Lemma~\ref{lemma:squeeze},
\[
	||\hat{\mu}^{TN}_{\hat{t}_F} - \mu||^2_2 \leq \max\{||\hat{\mu}^{HT}_{\hat{t}_F} - \mu||^2_2, ||\hat{\mu}^{ST}_{\hat{t}_F} - \mu||^2_2\} \leq ||\hat{\mu}^{HT}_{\hat{t}_F} - \mu||^2_2 + ||\hat{\mu}^{ST}_{\hat{t}_F} - \mu||^2_2
\]
we have
\begin{equation} \label{eq:risk_bound2}
E||\hat{\mu}^{TN}_{\hat{t}_F} - \mu||^2_2 \leq E||\hat{\mu}^{HT}_{\hat{t}_F} - \mu||^2_2 + E||\hat{\mu}^{ST}_{\hat{t}_F} - \mu||^2_2
\end{equation}

We can leverage the results of \citet{AbramBenDonJohn} and \cite{JiangZhangFDRST} and state the following theorem, essentially stated as they do in the in the former paper, but tweaked slightly to apply to our estimator:
\begin{theorem}\label{theo:risk}
		Let $y \sim N_n(\mu, I)$. Consider the estimator $\hat{\mu}^{TN}_{\hat{t}_F}$ applied with an FDR control parameter $q_n$, which may depend on $n$, but has some limit $q \in [0, 1)$. In addition, suppose $q_n \geq \gamma/\log(n)$ for some $\gamma > 0$ and all $n \geq 1$.
		
		Use the $\ell_r$ risk measure (\ref{eq:worst_case_risk}) where $0 \leq p < r \leq 2$. Let $\Theta_n$ be one of the parameter spaces detailed above with $\eta^p_n \in [\frac{\log^5(n)}{n}, n^{-\delta}]$, $\delta > 0$. Then as $n \rightarrow \infty$.
		\[	
			\sup_{\mu \in \Theta_n} \rho(\hat{\mu}^{TN}_{\hat{t}_F}, \mu) \leq 2^{w_{rp}}R_n(\Theta_n)\left[v_{rp} + u_{rp}\frac{(2q - 1)_+}{1 - q} + o(1)\right]
		\]
		where $w_{rp} = (r-1)_+$ for $0 \leq r < 2$ and $w_{2p} = 0$. Furthermore, $v_{rp} = 1 + \frac{u_{rp}}{1-q}$ for $0\leq r < 2$ and $v_{2p} = 2$ with $u_{rp} = 1$ and $u_{rp} = 1 - (p/r)$ for strong and weak $\ell_p$ balls respectively and $(x)_+ = \max\{x, 0\}$.
\end{theorem}

The theorem follows fairly simply from the bounds in (\ref{eq:risk_bound}) and (\ref{eq:risk_bound2}) and the analysis in \citet{AbramBenDonJohn} and \cite{JiangZhangFDRST}. The full discussion and proof is deferred to Appendix~\ref{appendix:proof_theo_risk}. An important lemma necessary for the lower bound is
\begin{lemma}\label{lemma:squeeze}
Consider the single component estimators $\hat{\mu}^{HT}_t(y)$, $\hat{\mu}^{ST}_t(y)$ and $\hat{\mu}^{TN}_t(y)$. Then,
		\[
			|\hat{\mu}^{HT}_t(y)| \geq |\hat{\mu}^{TN}_t(y)| \geq |\hat{\mu}^{ST}_t(y)|
		\]
		for all $y$, $t \geq 0$.
\end{lemma}
The proof of this lemma is detailed in Appendix~\ref{appendix:proof_lemma_squeeze}. It states that the value of the truncated Gaussian estimator is always squeezed in-between those of the hard and soft thresholding estimators. See Figure~\ref{fig:squeeze}. This allows us to bound the absolute difference between the TN and HT estimators above by the absolute difference in the HT and ST estimators, which is just $t$, the value of the threshold.

\begin{figure}[htb]
		\centering
		\includegraphics[width=80mm]{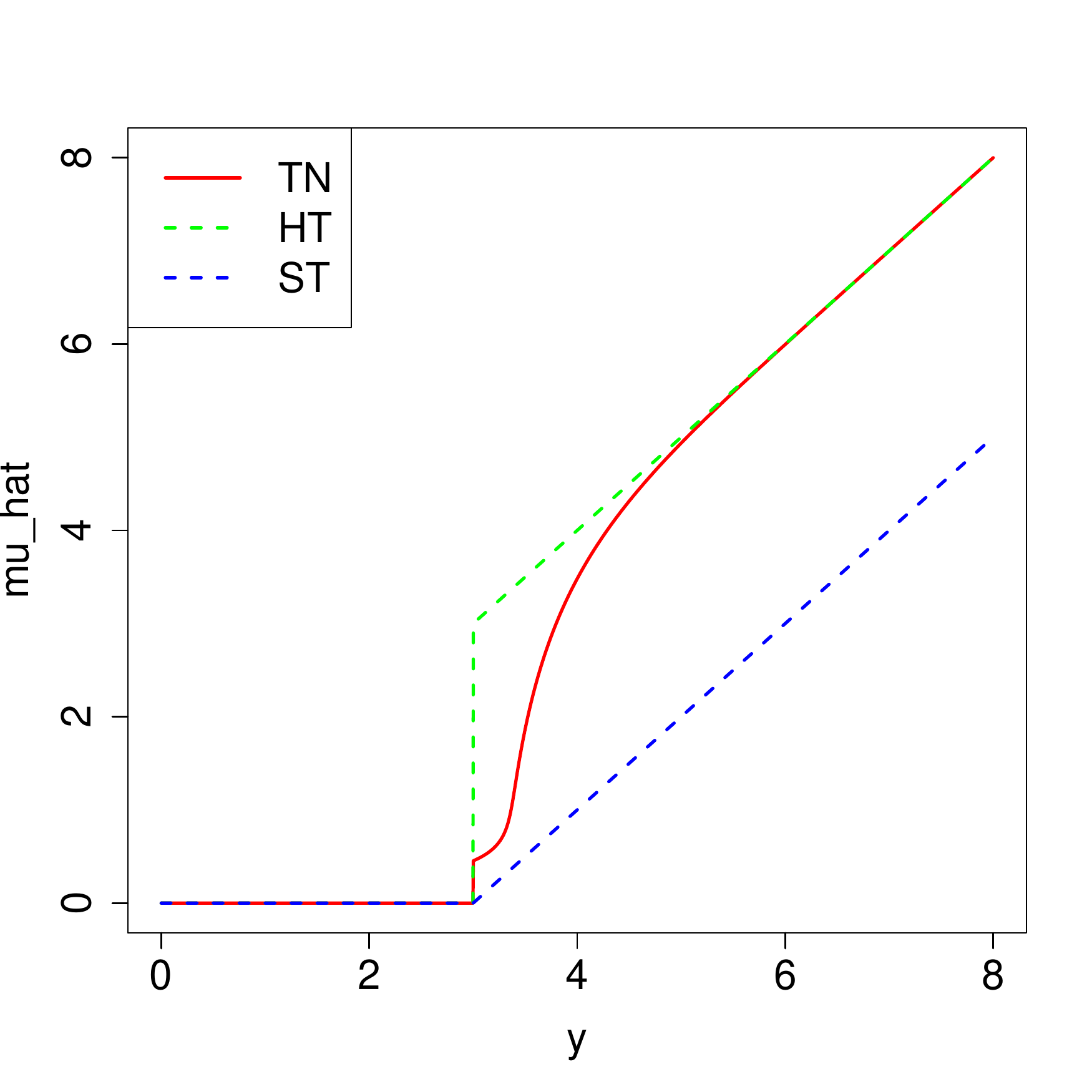}
		\caption{\emph{Comparison of the HT, ST and TN estimators. Horizontal axis measures the size of the sample element $y$, while the vertical axis plots the estimator value $\hat{\mu}$. Each is evaluated at the threshold value $t = 3$. Notice how the TN estimator lies everywhere between the HT and ST estimators. This holds for all $y$ and $t$, as is proven in Lemma~\ref{lemma:squeeze}.}}
		\label{fig:squeeze}
\end{figure}

We note that the bound is not tight. There is considerable scope for improvement. However, we have shown that the worst case risk of our estimator over sparsity balls is asymptotically bounded above by some constant of the minimax risk. This, coupled with the decent finite sample, partial MSE performance as exhibited in Section~\ref{sec:simulation}, makes the use of our estimator rather compelling. We note that the behaviour of the hard and soft thresholding estimators, for which we have tight asymptotic results, is rather poor compared to our estimator. We believe the upper bound of Theorem~\ref{theo:risk} can be tightened, probably to the degree of the bound in \citet{AbramBenDonJohn} and \cite{JiangZhangFDRST}, especially since our estimator is squeezed between the two estimators analyzed there. This endeavour, however, requires careful attention to minutiae and is postponed to future work.


\section{Confidence Intervals}\label{sec:ci}
We need not content ourselves merely with point estimates of the underlying signals. The post-selection, truncated Gaussian framework is quite readily extended to produce confidence intervals (CIs). Indeed, this seems to be the original intention of \citet{LeeSun2TaylorPostSel}.

To remind the reader: we apply some selection procedure (say $BH(q)$), which determines the set of variables, $E$, earmarked for further consideration. Previous sections dealt with obtaining estimates $\hat{\mu}_i$ for the signal $\mu_i$ underlying $y_i$, $i \in E$. Now we wish to obtain confidence intervals $[L^i_p, R^i_p]$ for these quantities such that:
\[
P(\mu_i \in [L^i_p, R^i_p]\mid S) = 1-p
\]
where the set $S$ contains all the necessary conditioning information (selected index set, signs, etc.) for the selection procedure of interest. Furthermore, we would like to ensure False Coverage Rate (FCR) control as defined by \citet{BY2005} for the family of intervals $\left\{ [L^i_p, R^i_p]\right\}_{j \in E}$. The FCR of this family is be limited to $p$.

In keeping with the results of \citet{LeeSun2TaylorPostSel}, we take, for each $i \in E$, $L^i_p$ and $R_p^i$ to satisfy:
\[
F^{[\mathcal{V}^-, \mathcal{V}^+]}_{L^i_p, \sigma^2}(y_i) = 1 - \frac{p}{2}
\]
and
\begin{equation}
F^{[\mathcal{V}^-, \mathcal{V}^+]}_{R^i_p, \sigma^2}(y_i) = \frac{p}{2}
\label{eqTNCIlimits}
\end{equation}
This guarantees the individual coverage probabilities \textit{conditional on the selection procedure} and the FCR are $1-p$ and bounded above by $p$ respectively.

FCR control has become the gold standard in post-selection, simultaneous confidence interval performance. It is not a panacea, however: one can have FCR control guarantees with confidence intervals still exhibiting poor behaviour. This is detailed in the next section.

\subsection{A possible pitfall of FCR control}

\citet{BY2005} proposed the following family of intervals for the selected signals: for each $i \in E$, construct the interval:

\[
[y_i - \sigma z_{\frac{p}{2|E|}}, y_i + \sigma z_{\frac{p}{2|E|}}]
\]
This construction ensures FCR control over signals selected in $E$. 

\citet{EfronLargeScaleBook} takes issue with this construction. He runs a fairly simple simulation, concisely revealing some of the pitfalls of blind adherence to the FCR control mantra. In his simulation, he generates $n = 10000$ replications of $y_i \sim N(\mu_i, 1)$ with $\mu_i = 0$ for $i = 1001, 1002, \dots, 10000$ (so 9000 noise variables). The 1000 signals $\mu_i$ for $i = 1, 2, \dots, 1000$ are generated independently from $N(-3, 1)$.

\begin{figure}[htb]
	\centering
	\includegraphics[width = 140mm]{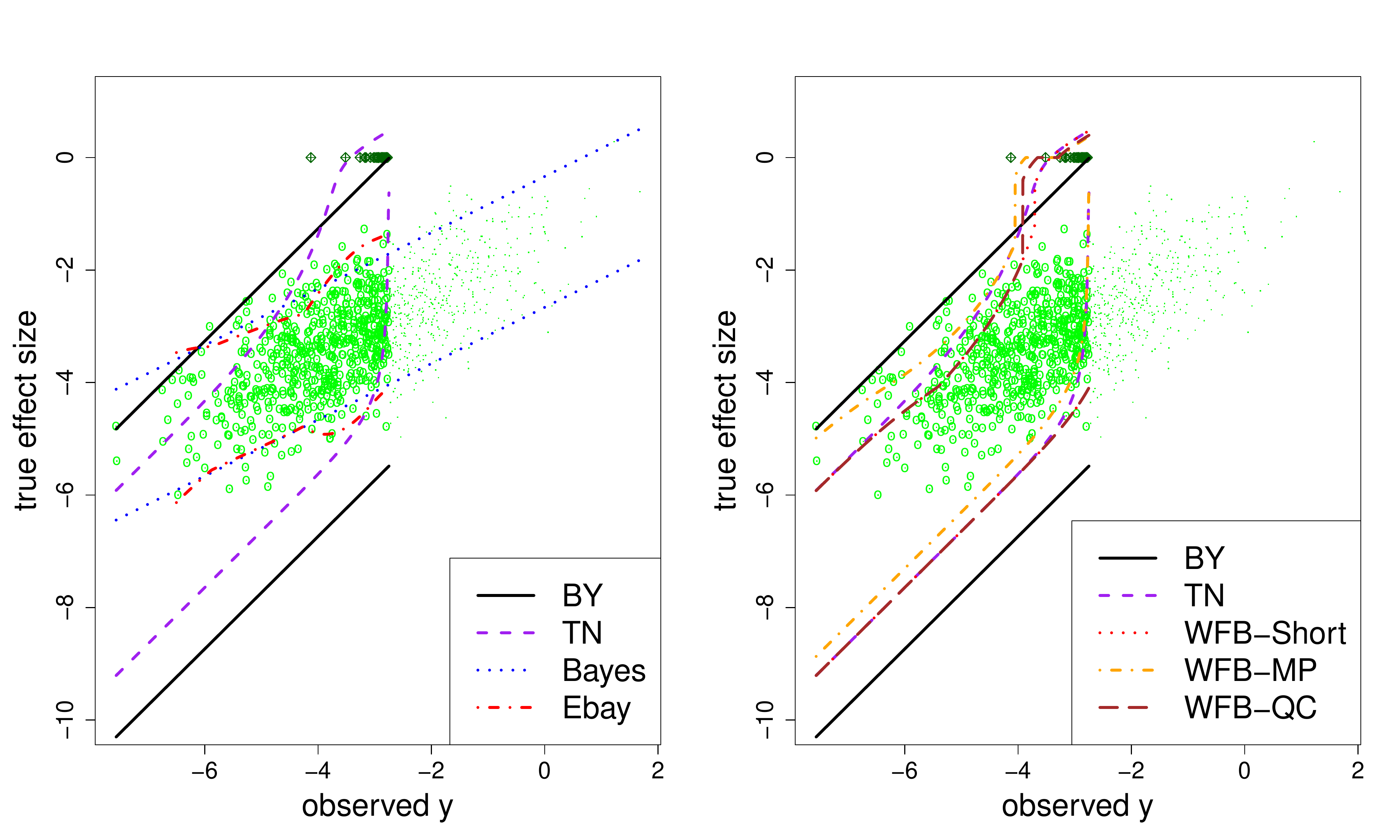}
	\caption{\emph{\textbf{Left panel}: True signal versus observed sample element for the Efron simulation. Green dots represent 1000 true signals; circled green dots, those signals selected by the $BH(0.1)$ procedure. Dark green diamonds represent false discoveries. Solid black lines delimit the Benjamini-Yekutieli (BY) post-selection intervals; purple broken lines, those obtained from the truncated Gaussian distribution (Equation \eqref{eqTNCIlimits}); red broken line, the empirical Bayes estimate of Efron; blue dotted lines, the true posterior Bayes prediction interval for the signals. All FCR guarantees at $p = 0.1$. \textbf{Right panel}: same setup, with black and purple lines as before. Red, orange and brown lines represent the three Weinstein-Fithian-Benjamini interval types.}}
	\label{figEfronCIs}
\end{figure}

The left panel of Figure~\ref{figEfronCIs} reproduces Figure 11.5 in \citet{EfronLargeScaleBook}. Benjamini-Yekutieli (BY) intervals are represented by black lines. As discussed in \citet{EfronLargeScaleBook}, these intervals seem to attain their FCR control guarantee partly by merit of their excessive width. Furthermore, notice how the cloud of selected signals (circled green dots) seem to lie toward the upper half of the BY intervals. Also, notice how the false discoveries (dark green diamonds) are not one covered by these intervals. Clearly then, there is no post-selection feedback giving us any notion of which could be the false rejections. FCR control is comforting, but we pay the price for it with very wide intervals that seem to be poorly centered around the selected signals.

Intervals based on the truncated Gaussian distribution (as in Equation~\eqref{eqTNCIlimits}) are plotted in purple. Notice how these are much narrower than the BY intervals and have non-uniform width.

Since the truncated Gaussian estimates are gleaned from maximizing a (conditional) likelihood of an exponential family, we can use the well-developed asymptotic theory of these types of estimators to construct the standard approximate confidence intervals:

\[
[\hat{\mu}_i - z_{\frac{p}{2}} \frac{1}{\sqrt{I_i(\hat{\mu}_i)}}, \hat{\mu}_i + z_{\frac{p}{2}} \frac{1}{\sqrt{I_i(\hat{\mu}_i)}}]
\]
where $\hat{\mu}_i$ is derived as in Equation~\eqref{eqCondMLE} and $I_i(\mu)$ is the Fisher information associated with the truncated Gaussian likelihood for observation $i$, evaluated at $\mu$. These intervals were found to have very similar shapes to those of the purple curves in Figure~\ref{figEfronCIs}, but tended to be more unstable, especially close to the truncation threshold. We note the existence of these intervals, but neither show them, nor consider them further.

Dotted blue lines show the true Bayesian posterior prediction interval for the signals. The prior on these signals is $N(-3, 1)$, the likelihood of the data $y$, given the signal is $N(\mu, 1)$, making the posterior $N(\frac{y+3}{2}, \frac{1}{2})$. The symmetric 90$\%$ probability interval for this distribution is shown in blue. Notice that neither the BY nor the truncated Gaussian (TN) intervals have the shape of the Bayesian posterior intervals. The former two types of interval are constructed post selection (hence do not see all the signals) and cannot hope to have the shape of the latter, which is privy to the entire distribution. If FCR control (with well-behaved intervals) is all we are after, we need not endeavor to replicate the shape of the posterior Bayes intervals exactly. All we require is good behavior for those signals we have selected.

Efron's empirical Bayes intervals are in red. These have the form
\begin{equation}
[{\rm E}_1(y) - \sqrt{{\rm Var}_1(y)}z_{\frac{p}{2}}, {\rm E}_1(y) + \sqrt{{\rm Var}_1(y)}z_{\frac{p}{2}}]
\label{eqEfronCI}
\end{equation}
where ${\rm E}_1(y) = \frac{y + l^\prime(y)}{1 - {\rm fdr}(y)}$, ${\rm Var}_1(y) = \frac{1 + l^{\prime\prime}(y)}{1 - {\rm fdr}(y)} - {\rm fdr}(y){\rm E}_1(y)^2$, ${\rm fdr}(y) = P(\mu = 0\mid y) = \frac{0.9\phi(y)}{f(y)}$ and $l(y) = \log{f(y)}$ with $f(y) $ the marginal distribution of the observed data. These formulae are obtained quite easily from an application of Bayes' and Tweedie's formulae. His proposed estimate of $f(y)$ is obtained via Lindsay's method, which comprises of a binning of $y$ values, followed by a Poisson natural cubic spline regression (with 7 degrees of freedom) of the counts of these bins onto their midpoints. Hence, 
\begin{equation}
\hat{f}(y) = C\cdot\exp{[\sum_{j = 1}^7 \hat{\beta}_jN_j(y)]}
\label{eqEfronTweedieF}
\end{equation}
with $N_j(x)$ the $j^{th}$ natural cubic spline basis function in the point $x$ and $C$ ensures the density integrates to 1. Estimates to the function $l$ and its derivatives are now easily obtained as:
\[
\hat{l}(y) = \log(C) + \sum_{j = 1}^7\hat{\beta}_jN_j(y)
\]
\[
\hat{l}^\prime(y) = \sum_{j = 1}^7\hat{\beta}_jN^\prime_j(y)
\]
\[
\hat{l}^{\prime\prime}(y) = \sum_{j = 1}^7\hat{\beta}_jN^{\prime\prime}_j(y)
\]

These quantities are then fed into Equation~\eqref{eqEfronCI} to obtain estimates of the confidence intervals. The theory pertaining to these intervals is quite neat and anecdotal evidence (like plots similar to the left panel of Figure~\ref{figEfronCIs}) suggest that they may be very useful. They are derived from a method that tries to estimate the entire density of the data, allowing the estimated CIs to bend closer to the true posterior Bayes interval. We are not aware of FCR control guarantees on only the selected signals though. Furthermore, CIs estimated using the empirical Bayes methodology are very unstable and often lead to inconsistent numerical results. We present an example in Appendix~\ref{appendix:numer_instab}.

\subsection{Competitors to truncated Gaussian post-selection confidence intervals}

\citet{WFB2013} propose three types of intervals, also based on the truncated Gaussian distribution. Their method of construction is somewhat more involved and requires the inversion of the acceptance region of a particular hypothesis test each time. The three confidence intervals are constructed, respectively, from a test with shortest acceptance region (WFB-Short), a conditional modified Pratt procedure (WFB-MP) and a conditional, quasi-conventional interval with parameter trading off between power and interval length (WFB-QC). Details are to be found in their paper. 

Figure~\ref{figEfronCIs} (right panel) compares the intervals constructed by our method and the three of \citet{WFB2013} for the same replication of Efron's experiment as in the left panel of that plot. Weinstein-Fithian-Benjamini intervals were obtained using the default settings in the software provided with their manuscript. Intervals seem to be of similar shapes. However, differences seem significant enough (especially around -5 in the upper limits and around -2 in the lower limits) that a more extensive simulation study is called for. \citet{WFB2013} show how their intervals also have FDR control guarantees. The extended simulation study then will focus on other properties of the intervals: width, symmetry of coverage misses and symmetry of coverage hits.

\subsection{A simulation study}
Five methods of confidence interval construction were compared: Benjamini-Yekutieli intervals, our truncated Gaussian intervals and the three of \citet{WFB2013}. $S = 30$ replications of the Efron experiment were run and the interval types compared on three criteria (to be detailed later):
\begin{itemize}
	\item Interval width.
	\item Symmetry of coverage misses.
	\item Symmetry of coverage hits.
\end{itemize}

The reader is reminded that Efron's experiment generates a sample of size $n = 10000$ from $N(\mu_i, 1)$ with 1000 $\mu_i$ signals from a $N(\nu, 1)$ prior distribution and 9000 noise variables. The BH(0.1) procedure is then run to determine a set of variables of interest to further study. Post-selection intervals are then computed for these only, aiming for an FCR of 0.1. We looked at two signal strengths $\nu = -3$ and $\nu = -5$. Figures~\ref{figCISim3} and~\ref{figCISim5} summarize the main findings.

\begin{figure}[htb]
	\centering
	\includegraphics[width=150mm]{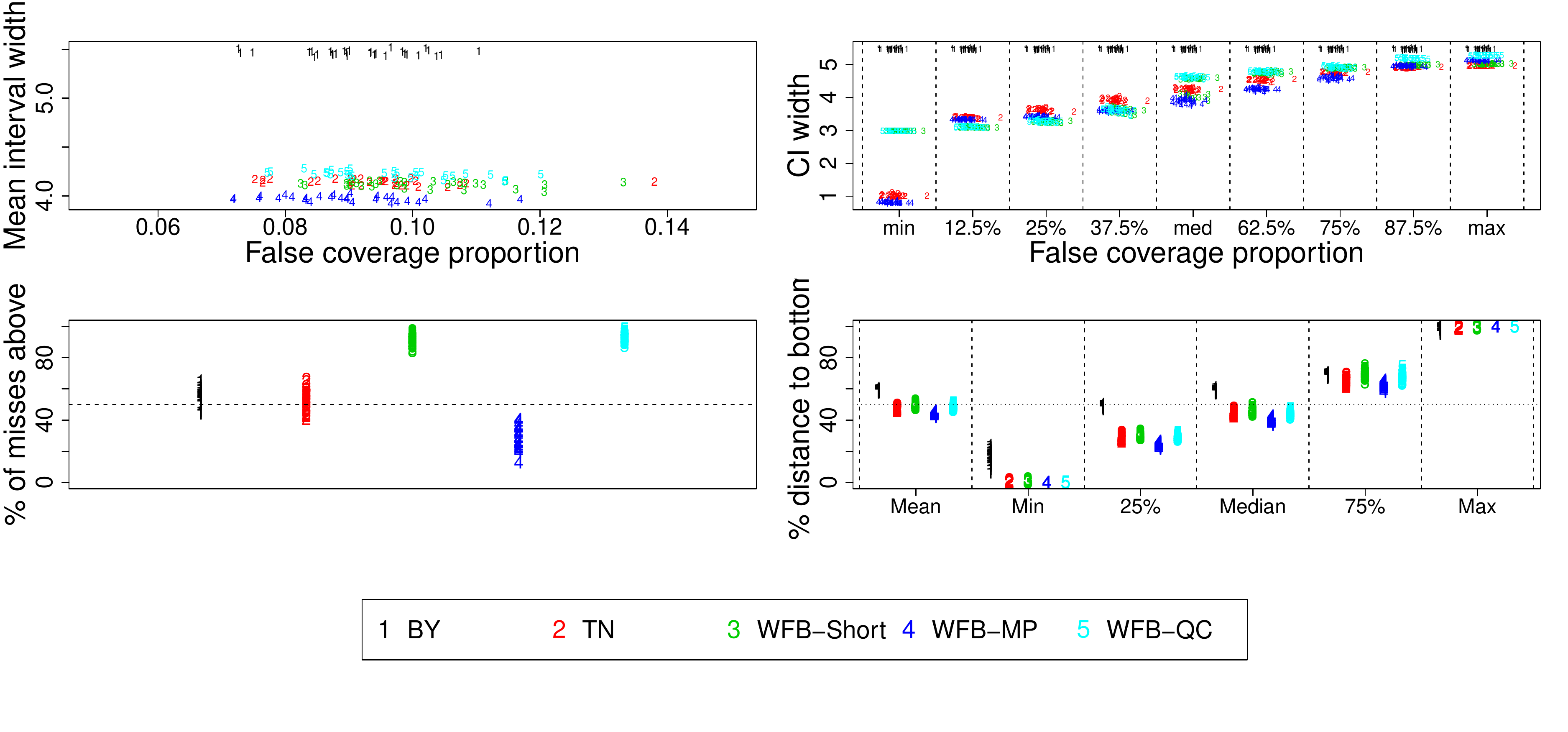}
	\caption{\emph{Signal size $\nu = -3$. \textbf{Top left:} Mean interval width for each replication versus false coverage proportion. 1: Benjamini-Yekutieli intervals, 2: truncated Gaussin intervals, 3: Weinstein-Fithian-Benjamini Shortest acceptance region intervals, 4: Weinstein-Fithian-Benjamini Modified Pratt procedure intervals, 5: Weinstein-Fithian-Benjamini: quasi-conventional intervals. \textbf{Top right:} A sequence of zooms on interval [0.05, 0.15], with each strip plotting a sample summary statistic of interval widths versus false coverage proportions for each of the replications. Strips plot minimum intervals widths, all the octiles and the maximum interval width. \textbf{Bottom left:} Proportion of intervals at each replication missing their true signal where the latter lies above the former. Horizontal axis has no interpretation. Black horizontal line at proportion 0.5, for reference. \textbf{Bottom right:} Minimum, maximum and quantiles of proportion of distance from true signal to lower bound to that of the entire interval width for each replication. Again, horizontal axis has no interpretation.  Black horizontal line at proportion 0.5, for reference.}}
	\label{figCISim3}
\end{figure}

\begin{figure}[htb]
	\centering
	\includegraphics[width=150mm]{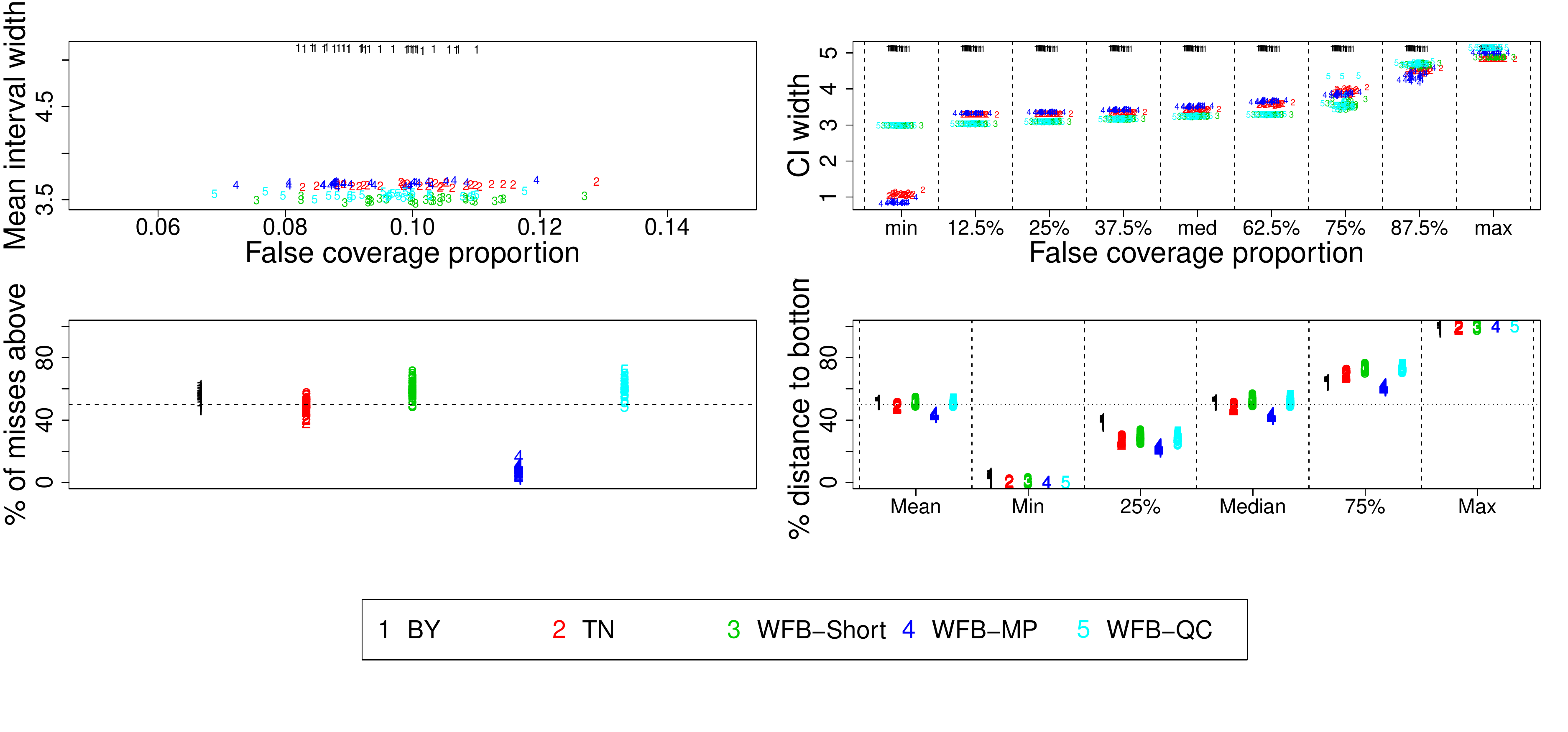}
	\caption{\emph{As for Figure~\ref{figCISim3}, but with signal size $\nu = -5$.}}
	\label{figCISim5}
\end{figure}

\subsection{Interval width}
The top rows of Figures~\ref{figCISim3} and~\ref{figCISim5} pertain to interval widths. For each replication of the experiment (dropping subscripts referring to the replication number), we identify $|E|$ items for further study. We construct an interval $[L^i_p, R^i_p]$ for each of these with $p = 0.1$. The left panel plots mean interval width
\[
\frac{1}{|E|}\sum_{i \in E}(R_p^i - L_p^i)
\]
against false coverage proportion
\[
\frac{1}{|E|}\sum_{i \in E}I\{\mu_i \not\in [L^i_p, R^i_p]\}
\]
for each method for each replication. No method's point cloud lies significantly to the right or left of those of the others. Also, all point clouds seem to cluster around 0.1 on the horizontal axis, as is expected since all methods have an false coverage rate control guarantee. Methods differ in the average interval width they construct. Benjamini-Yekutieli intervals are by far the widest, with Weinstein-Fithian-Benjamini-MP on average the narrowest. The other three methods seem to produce intervals of similar width.

The right panel of these figures makes similar plots, only for summary statistics other than the mean (namely, the minimum, maximum and all the octiles). Each strip in this panel (delimited by vertical dashed black lines) zooms in on the horizontal axis interval [0.05, 0.15]. The horizontal axis overall does not have its usual interpretation. This is just so that we could fit many different plots in one panel. Now we get a richer view of the intervals constructed by each method. The truncated Gaussian and Weinstein-Fithian-Benjamini-MP seem to construct intervals with the shortest minimum length. The truncated Gaussian's shortest intervals seem to be shorter than most, with the left middle of the interval width distribution perhaps to the right of the others (indicating wider intervals), but with slower growth of interval width, eventually resulting in the shortest maximum interval widths. Bear in mind that the methods share similar coverage. It is interesting to see how they achieve this coverage though. Benjamini-Yekutieli intervals are of constant length and much wider than the rest. Very little difference in interval width is encountered over different signal sizes (the top panels of the two figures look very similar).

Based solely on interval width, one would probably say that the Weinstein-Fithian-Benjamini-MP intervals present as marginally superior to the rest. However, there are other characteristics we require of intervals, one of which is that coverage misses be symmetric.

\subsection{Symmetry of coverage misses}

The bottom left panels of Figures~\ref{figCISim3} and~\ref{figCISim5} plot the proportion of upward coverage misses
\[
\frac{\sum_{i \in E}I\{\mu_i > R^i_p\}}{\sum_{i \in E}I\{\mu_i \not\in [L^i_p, R^i_p]\}}
\]
for each of the methods, for each replication. The horizontal axis has no interpretation. If a method produces, on average, intervals that are equally likely to miss the true signal above or below, we would expect this quantity to hover around 0.5. Here we see significant differences between Figures~\ref{figCISim3} and~\ref{figCISim5}.

For smaller signals (Figure~\ref{figCISim3}), we see that the truncated Gaussian method seems to be the only one producing intervals with roughly symmetric coverage (its points are neatly bisected by the horizontal reference line at 0.5). Benjamini-Yekutieli intervals tend to miss slightly more signals with the signals lying above the interval, while Weinstein-Fithian-Benjamini-Short and Weinstein-Fithian-Benjamini-QC making considerably more of these errors. Also, the shorter Weinstein-Fithian-Benjamini-MP intervals seem to miss many of the true signals above. Notice the improvement of all (except Weinstein-Fithian-Benjamini-MP) for larger signals (bottom left panel, Figure~\ref{figCISim5}).

\subsection{Symmetry of coverage hits}

We require not only that our intervals miss symmetrically, but also that those true signals covered generally find themselves in the middle of its constructed interval. The bottom right panels of Figures~\ref{figCISim3} and~\ref{figCISim5} plot the summary statistics of an interval skewness measure:
\[
{\rm skew}_i = \frac{\mu_i - L^i_p}{R^i_p - L^i_p}.
\]
Again, the horizontal axis has no interpretation other than to separate the methods. If true signals tend to find themselves in the middle of the intervals, we would expect the mean and median of these quantities (for each replication) to be around 0.5.

Indeed, for small signals (Figure~\ref{figCISim3}) we see that the means and medians (first and fourth strips of bottom right panel) seem to be roughly 0.5 for the truncated Gaussian, Weinstein-Fithian-Benjamini-Short and Weinstein-Fithian-Benjamini-QC intervals. Benjamini-Yekutieli intervals seem to place true signals closer to the upper than the lower limit (as we expect, considering Figure~\ref{figEfronCIs}), while the Weinstein-Fithian-Benjamini-MP interval tends to put true signals closer to the lower limit. Again, matters improve when signal sizes get larger.

Overall then, even for moderate signals, the truncated Gaussian method seems to produce intervals that compete in width with the shortest of the post-selection intervals, while maintaing better symmetry in both coverage misses and hits.

\section{Discussion}\label{sec:conclude}
We have applied the ideas of \citet{LeeSun2TaylorPostSel} to the orthogonal design setting, whereby we have a single sample of elements with different underlying signals. We wish to identify those with non-zero signals and then produce good estimates (both point and interval) of these after taking into account the selection bias of the winners' curse. It was shown that top-$K$ and Benjamini-Hochberg selection procedures also satisfy the linear constraint representation of \citet{LeeSun2TaylorPostSel}, allowing us to leverage many of their results.

The resulting estimates of post-selection signal size compete with the best known post-selection estimates for moderately sparse signals. Especially pleasing is the performance of our method after using a Benjamini-Hochberg procedure to select signals, with $q$ taking on practicable values. Furthermore, the confidence intervals produced by the method improve significantly on industry standard Benjamini-Yekutieli intervals, competing favorably against recently developed competitors.  

Another promising aspect of our estimator is that it can be shown to have worst case risk bounded within a constant multiple of minimax risk over a rich set of parameter spaces evoking sparsity of signal. This holds not only for squared error loss, but for losses of the form $||\hat{\mu} -\mu||^r_r$, $0 < r \leq 2$. We presented a bound and argued that it could most likely be tightened considerably.

\bibliographystyle{agsm}
\bibliography{mmBib}


\appendix
\section{Proof of Lemma~\ref{lemma:squeeze}} \label{appendix:proof_lemma_squeeze}

First we prove monotonicity of the derivatives of the mean functions of truncated Gaussian random variables:
	\begin{lemma}\label{lemma:monotone_derivative}
		Suppose $Y \sim N(\mu, 1)$. Now define for $i = 1, 2$, $E_i(\mu, t) = E[Y|Y\in I_i]$ with $I_1 = (-\infty, -t] \cup [t, \infty)$ and $I_2 = [-t, t]$. Then for $i = 1,2$:
		\[
			\frac{\partial}{\partial \mu}E_i(\mu, t) \geq 0
		\]
	\end{lemma}
	
	\begin{proof}
		Define $M_i(\mu, t) = E[Y^2| Y\in I_i]$ and  $Var_i(\mu, t) = {\rm Var}[Y|Y\in I_i] = M_i(\mu, t) - E_i(\mu, t)^2$ and note that $E_i(\mu, t) = \mu + h_i(\mu, t)$ where
		\[
			h_1(\mu, t) = \frac{\phi(t-\mu) - \phi(-t-\mu)}{1 - \Phi(t-\mu) + \Phi(-t-\mu)} = \frac{\phi(t-\mu) - \phi(-t-\mu)}{Z_1} (say)
		\]
		and
		\[
			h_2(\mu, t) = \frac{\phi(-t-\mu) - \phi(t-\mu)}{\Phi(t-\mu) - \Phi(-t-\mu)} = \frac{\phi(-t-\mu) - \phi(t-\mu)}{Z_2} (say)
		\]
		Now
		\[
			\frac{\partial}{\partial \mu}\frac{1}{Z_i} = -\frac{h_i(\mu, t)}{Z_i}
		\]
		so that
		\begin{align*}
			\frac{\partial}{\partial \mu}E_i(\mu, t) &= \frac{\partial}{\partial \mu} \int_{-t}^t x\frac{\phi(x-\mu)}{Z_i}\,dx \\
			&= \int_{-t}^t x(x-\mu)\frac{\phi(x-\mu)}{Z_i}\,dx - h_i(\mu, t)\int_{-t}^tx\frac{\phi(x-\mu)}{Z_i}\,dx \\
			&= M_i(\mu, t) - \left(\mu + h_i(\mu, t)\right)E_i(\mu, t) = M_i(\mu, t) - E_i(\mu, t)^2 = Var_i(\mu, t) \geq 0
		\end{align*}
	\end{proof}
	
	We use this lemma to prove Lemma~\ref{lemma:squeeze}
	\begin{proof}[Proof of Lemma~\ref{lemma:squeeze}]
		We prove the case where $y > t$, since the other side follows by symmetry. Suppose $y > t > 0$. We know by differentiating and setting to zero that $\hat{\mu}^{TN}_t$ satisfies $E_1(\hat{\mu}^{TN}_t, t) = y$. Furthermore, from Lemma~\ref{lemma:monotone_derivative}, this mean function is monotone and increasing. Thus, our stated result will follow if $E_1(\hat{\mu}^{HT}_t, t) \geq y$ and  $E_1(\hat{\mu}^{ST}_t, t) \leq y$.
		
		If $y > t$, then $\hat{\mu}^{HT}_t = y$ and 
		\[
			E_1(\hat{\mu}^{HT}_t, t) = y + \frac{\phi(y - t) - \phi(y + t)}{1 - \Phi(t - y) + \Phi(-t - y)} \geq y
		\]
		since $y + t > y - t > 0$ and $\phi(\cdot)$ is decreasing on the positive real line.
		
		Similarly, if $y > t$ then $\hat{\mu}^{ST}_t = y - t = \epsilon$ (say) and
		\begin{align*}
			E_1(\hat{\mu}^{ST}_t, t) &= y - t + \frac{\phi(t - \epsilon) - \phi(t + \epsilon)}{1 - \Phi(t - \epsilon) + \Phi(-t - \epsilon)} \\
			&= y - t + \frac{\phi(\epsilon - t) - \phi(-\epsilon - t)}{2\Phi(-\epsilon - t)+\Phi(\epsilon - t) - \Phi(-\epsilon - t)} \\
			&\leq y - t + \frac{\phi(\epsilon - t) - \phi(-\epsilon - t)}{\Phi(\epsilon - t) - \Phi(-\epsilon - t)}\\
			&= y - E_2(t, \epsilon)
		\end{align*}
		Integrating both sides of the result of the lemma above, we get
		\[
			E_2(t,\epsilon) \geq E_2(0, \epsilon) = 0
		\]
		by the symmetry of $\phi(\cdot)$. Hence $E_1(\hat{\mu}^{ST}_t, t) \leq y$, as required.
		
	\end{proof}


\section{Proof of Theorem~\ref{theo:risk}}\label{appendix:proof_theo_risk}

The proof of our Theorem~\ref{theo:risk} is a simple application of the results of \citet{AbramBenDonJohn}. Their paper provides a very thorough and incisive analysis of the quantities of interest and we leverage their results to a great extent. We collect some of their notation and results here for ease of reference. These are used extensively in what follows, especially in Appendix~\ref{appendix:soft_thresh}.
		
		\citet{AbramBenDonJohn} note that $\hat{k}_F$ is the leftmost local minimum of $S_k^{(r)} = \sum_{l = 1}^k t_l^r + \sum_{l = k+1}^n |y|_{(l)}^r$ for $ 0 < r \leq 2$, which, with a little thought, can be seen to lead the the index selected by the $BH(q)$ as the last rejection, or smallest non-zero signal. The minimax risk over parameter space $\Theta_n$ is denoted as $R_n(\Theta_n) = \inf_{\hat{\mu}}\tilde{\rho}(\hat{\mu}, \Theta_n)$. Asymptotic minimax risks for each of the parameter spaces of interest is given by
		\begin{enumerate}
			\item $R_n(\ell_0[\eta_n]) \sim n\eta_n\tau_\eta^r$
			\item $R_n(m_p[\eta_n]) \sim \frac{r}{r-p} R_n(\ell_p[\eta_n])$
			\item $R_n(\ell_p[\eta_n]) \sim n\eta_n\tau_\tau^{r-p}$
		\end{enumerate}
		where $\tau_\eta = \sqrt{2\log\eta_n^{-p}}$
		
		They say that an event $A_n(\mu)$ is $\Theta_n$-likely if there exist constants $c_0$ and $c_1$, not depending on $n$ and $\Theta_n$ such that
		\[
			\sum_{\mu \in \Theta_n}P_\mu(A^c_n(\mu)) \leq c_0\exp(-c_1\log^2n)
		\]
		
		Furthermore, they define $\alpha_n = 1/b_4\tau_\eta$ with $b_4 = (1-q)/4$. They go on to define many different indices. We provide a list:
		\begin{itemize}
			\item $k(\mu) = \inf\{k \in \Re^+: \sum_{l = 1}^nP_\mu(|y_l| > t_k) = k\}$
			\item $\hat{k}_F$, $\hat{k}_G$ and $\hat{k}_r$ are the leftmost local, rightmost local and global minima of $S^{(r)}_k$.
			\item $k_n = [n\eta_n]$ for $\Theta_n = \ell_0[\eta_n]$ and $k_n = n\eta_n^p\tau_\eta^{-p}$ for $\Theta_n = \ell_p[\eta_n]$ and $\Theta_n = m_p[\eta_n]$.
			\item $k_-(\mu) = (k(\mu) - \alpha_nk_n)I\{k(\mu) > 2\alpha_nk_n\}$
			\item $k_+(\mu) = \alpha_nk_n + \max\{k(\mu), \alpha_nk_n\}$
			\item $\kappa_n = [\alpha_n + (1-q_n)^{-1}]k_n$ for $\ell_0[\eta_n]$ and $\kappa_n = [\alpha_n + (1-d_n-q_n)^{-1}]k_n$ for $m_p[\eta_n]$ with $d_n = 2c_0\tau_\eta^{-1}$ ans $c_0$ a constant defined in the paper.
		\end{itemize}
		
		The authors show that $\sup_{\mu \in \Theta_n}k_+(\mu) \leq \kappa_n$. Also note that $k_+(\mu) - k_-(\mu) \leq 3\alpha_nk_n$. The authors show that it is $\Theta_n$-likely on the parameter spaces of interest that $k_-(\mu) \leq \hat{k}_F \leq \hat{k}_G \leq k_+(\mu)$. Furthermore, they prove the monotonicity of $k \rightarrow kt_{k}^r$ for $0 \leq r \leq 2$, so that it is $\Theta_n$-likely that
		\begin{equation}\label{eq:htst_diff}
			\hat{k}_Ft_{\hat{k}_F}^r \leq k_+(\mu)t_{k_+(\mu)}^r \leq \kappa_nt_{\kappa_n}^r = \frac{u_{rp}}{1-q_n}R_n(\Theta_n)(1 + o(1))
		\end{equation}
		where $q_n$ is a sequence satisfying the assumptions in the statement of Theorem~\ref{theo:risk}. The second term of (\ref{eq:risk_bound}) is analysed in \citet{AbramBenDonJohn} under the assumptions stated in the theorem. Combining their result with our (\ref{eq:htst_diff}) completes the proof of our Theorem~\ref{theo:risk} for $0 \leq r < 2$.

We can use the results of Theorem 4 in \citet{JiangZhangFDRST} to tighten the bound for $r = 2$ (squared error loss). In the reference they prove that
\[
\sup_{\mu \in \Theta_n} \rho(\hat{\mu}^{ST}_{\hat{t}_F}, \mu) \leq R_n(\Theta_n)(1 + o(1))
\]
for $(L_{0, n}/n^{1+\delta_{1, n}})^{p/2}/n \ll \eta_n^p \ll 1$ where $L_{0, n} = (\log n)^{-3/2}\left(\delta_{1,n}(\log n)^{3/2} + (\log \log n)\sqrt{\log n}\right)^{1+\delta_{1, n}}$ and $\delta_{1, n} \rightarrow 0$ when $0 < p \leq 2$ and $n^{-1} \ll \eta_n \ll 1$ for $p = 0$ (the nearly black case). Note that when we have $\eta_n^p \in [\log^5 n, n^{-\delta}]$ then we meet the requirement and we can use the bound in~(\ref{eq:risk_bound2}) to conclude the (tighter) theorem statement for $r = 2$.


\section{Numerical instability of Empirical Bayes CIs}\label{appendix:numer_instab}
Although the theory is sound and rather intuitive, numerical instability makes it very difficult to apply the empirical Bayes framework to the construction of confidence intervals. 

\begin{figure}[htb]
	\centering
	\includegraphics[width=100mm]{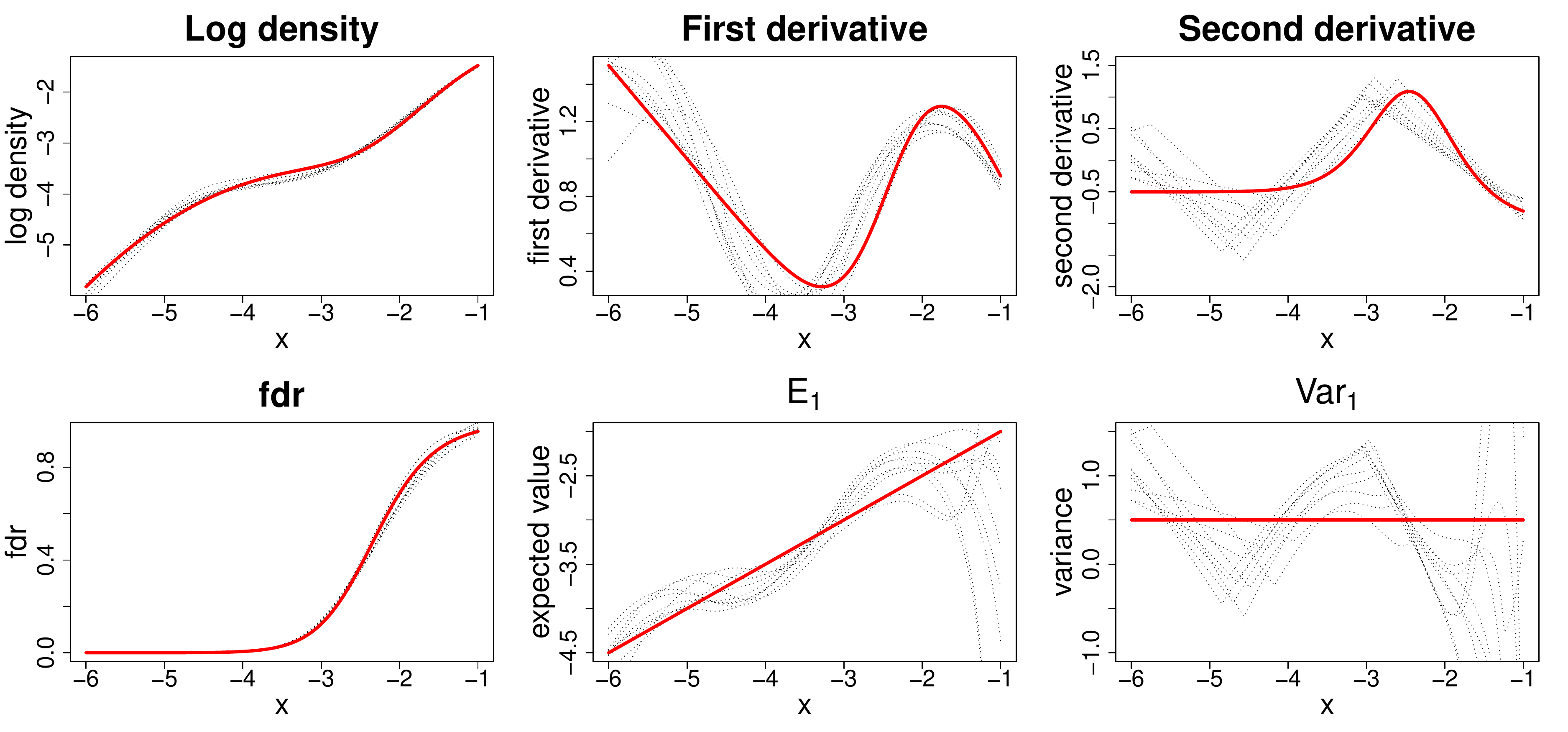}
	\caption{\emph{\textbf{Top row:} Estimates of the log density and its first two derivatives (over the left half of the distribution) using the empirical Bayes framework and Poisson regression density estimate of Efron. This is for the Efron simulation setup detailed in Section 6.1. Thick red lines show the true function. \textbf{Bottom row:} Estimates of local FDR, posterior expectation of the signals and their variance based on the density estimates and the formulae below Equation~\ref{eqEfronCI}. Again, red curves show the truth.}}
	\label{figEfronDensity}
\end{figure}

Figure~\ref{figEfronDensity} shows the density estimates from $S = 20$ replications of the Efron simulation of Figure~\ref{figEfronCIs}. Notice how well the log density is estimated. Local FDR (bottom left) is also well estimated, because it depends only on the estimated density. Notice, however, how estimates of density derivatives (and quantities based on them) seem to become progressively worse at higher orders. Estimates of the second derivative (top right) look decidedly ropey. This has dire implications for the estimated posterior variance of our signals (just below it, bottom right). Notice how often we underestimate the true variance (which is constant at 0.5). Variance estimates quite often become negative as well, leading to utterly invalid intervals. This is not a fault of the theory, but stems rather from the numerical instability of estimation of these higher order derivatives.

It should be noted that the numerical instability is not only the fault of the application of Lindsay's method, but seems rather to be a result of the inherent difficulty of estimating a function and its first two derivatives and the shortcomings of the tools we have to do so. Figure~\ref{figWagerDensity} shows the non-linear programming density estimates of \citet{WagerDens} for the same replications of the Efron experiment. These density estimates have been shown to outperform just about all other competitors. Second order derivative estimates seem more stable, but notice they are still lower than they should be at many places and still lead to negative posterior variances. Perhaps some effort should be spent on finding a new class of density estimators, with specific focus on regularising the second derivative estimate. Otherwise empirical Bayes methods are lost to the successful construction of post-selection intervals. These ideas are not pursued further here. 

\begin{figure}
	\centering
	\includegraphics[width=100mm]{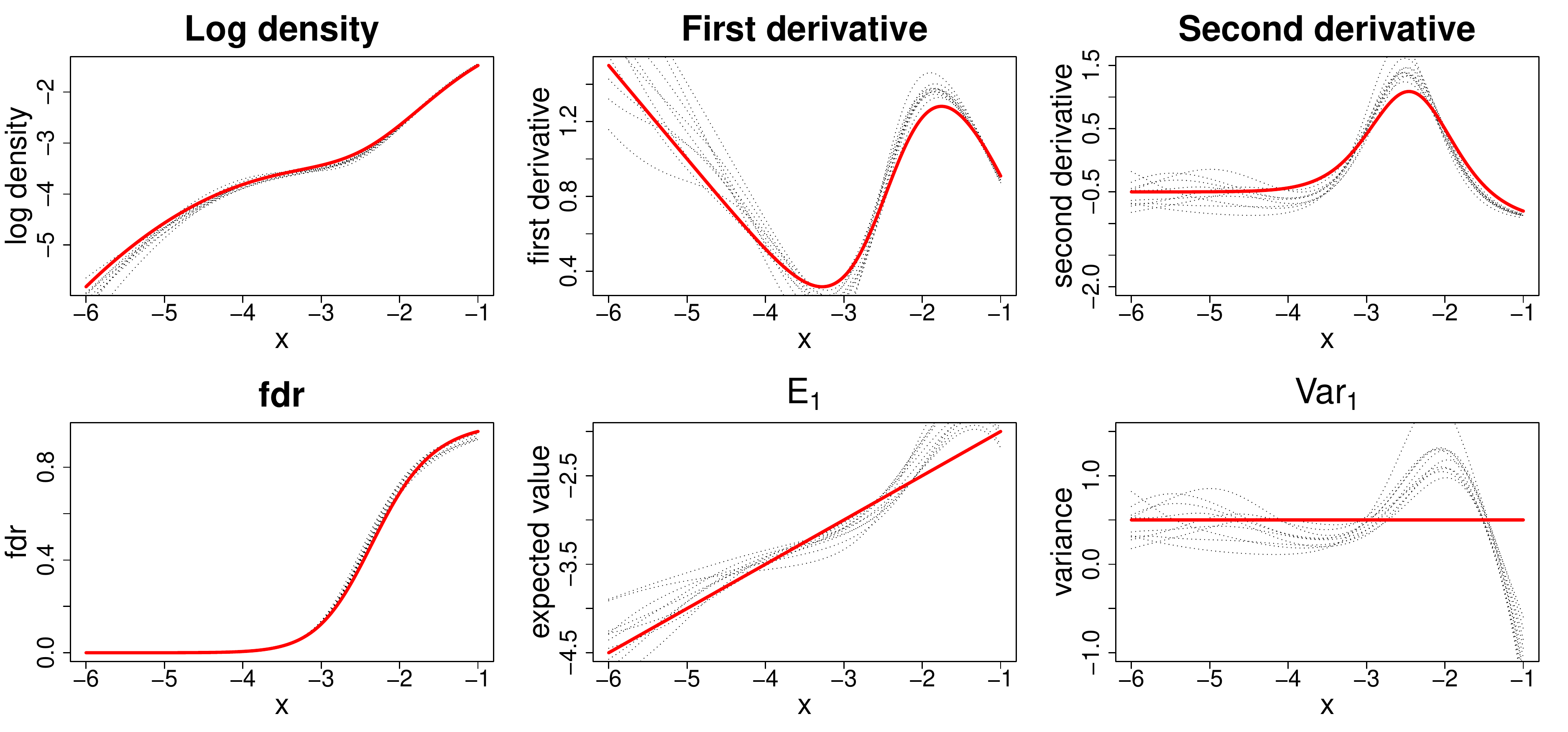}
	\caption{\emph{As for Figure~\ref{figEfronDensity}, but for the empirical Bayes, non-linear programming density estimate of Wager.}}
	\label{figWagerDensity}
\end{figure}

\end{document}